\documentclass[10pt,fleqn,aps,prb,twocolumn,showpacs,floatfix]{revtex4-2}
\usepackage[a4paper,margin=1in]{geometry}
\usepackage{amsmath,amssymb,amsthm,mathtools,bm}
\usepackage{mathrsfs}
\usepackage{tensor}
\usepackage{bbm}
\usepackage{graphicx}
\usepackage{url}
\usepackage[utf8]{inputenc}
\usepackage{subfigure}
\usepackage{slashed}
\usepackage{dsfont}
\usepackage{setspace}
\usepackage{wasysym}
\usepackage{makecell}
\usepackage{xcolor}
\usepackage{ulem}
\usepackage{hyperref}
\usepackage[shortlabels]{enumitem}

\newcommand{\Id}{\mathbbm{1}}

\newcommand{\normop}[1]{\left\lVert #1 \right\rVert}

\newcommand{\ket}[1]{|#1\rangle}
\newcommand{\bra}[1]{\langle #1|}
\newcommand{\braket}[2]{\langle #1 | #2 \rangle}

\newcommand{\tr}{\operatorname{tr}}

\newcommand{\I}{\mathbbm{1}}

\newcommand{\cN}{\mathcal{N}}

\newcommand{\eps}{\epsilon}

\theoremstyle{plain}
\newtheorem{proposition}{Proposition}
\newtheorem{lemma}{Lemma}
\newtheorem{corollary}{Corollary}

\theoremstyle{definition}
\newtheorem{assumption}{Assumption}

\theoremstyle{remark}
\newtheorem{remark}{Remark}

\newcommand{\be}{\begin{equation}}
\newcommand{\ee}{\end{equation}}
\newcommand{\bea}{\begin{eqnarray}}
\newcommand{\eea}{\end{eqnarray}}

\begin{document}

\title{Hierarchical entanglement transitions and hidden area-law sectors in quantum many-body dynamics}

\author{Tarun Grover}
\affiliation{Department of Physics, University of California at San Diego, La Jolla, California 92093, USA}
\begin{abstract}
Chaotic many-body dynamics typically generates volume-law entanglement from initially low-entangled states. 
We reveal an intricate, hierarchical entanglement structure in local quantum quenches, both in the canonical purification of locally quenched Gibbs states and in a companion pure-state circuit model. 
In either setting, the full state exhibits a Rényi-index-tuned transition: at long times, \(S_{\alpha>1}\) obeys an area law, while \(S_{\alpha\le 1}\) is volume-law. 
More strikingly, the response linear in the quench strength is carried by only an \(O(1)\)-dimensional dominant Schmidt sector; the corresponding states exhibit their own area-to-volume-law transitions at critical indices \(\alpha_c<1\), implying polynomial-bond-dimension approximability in one dimension. We provide evidence that this hierarchy persists recursively: upon bipartitioning the dominant Schmidt states, their leading Schmidt sectors exhibit analogous structure. We derive the mechanism analytically in the circuit model, prove the \(S_{\alpha>1}\) area law for locally quenched Gibbs states, and support the hierarchy by exact diagonalization of random circuits and locally quenched Gibbs states of chaotic spin chains.
\end{abstract}

\maketitle

\textbf{Introduction.}  The scaling of entanglement entropy provides a useful coarse diagnostic of quantum many-body states: ground states of local Hamiltonians typically satisfy an area law up to possible logarithmic corrections~\cite{srednicki1993,Holzhey94,Calabrese04,hastings2007area}, whereas typical eigenstates at finite energy density, as well as states evolved for time linear in system size, exhibit a volume law~\cite{deutsch1991,srednicki1994chaos,page_1993,bravyi2006lieb,calabrese2005evolution}. This expectation explains both the success of tensor-network methods for many one-dimensional ground-state problems~\cite{White92,White93,verstraete2004density,VerstraeteCirac06,vidal2007classical,Schollwoeck11}  and the difficulty of generic long-time dynamics~\cite{Vidal04,daley2004time,white2004realtime,schuch2008entropy, Orus14,Schollwoeck11, karrasch2012finiteT,barthel2013precise}, despite recent progress with tensor-network and operator-space methods~\cite{leviatan2017quantum,Houschild2018finding,White2018quantumdynamics,Rakovszky2022dissipation,Keyserlingk2022Operator,YiThomas2024Comparing,Artiaco2024Efficient,Angrisani2025classically,rudolph2025pauli,cruz2025quantum,anand2026kpz}. In this work, we show that even in highly chaotic systems, local quantum quenches can hide a low-entanglement sector inside a highly entangled state. This sector contains the entire response linear in the quench strength, even though the full state can have volume-law entanglement. We further show that this is part of a broader hierarchical structure: as one recursively bipartitions the dominant Schmidt states, their entanglement exhibits area-to-volume-law transitions as the Rényi index is varied.

\begin{figure}[t]
\centering
\includegraphics[width=\columnwidth]{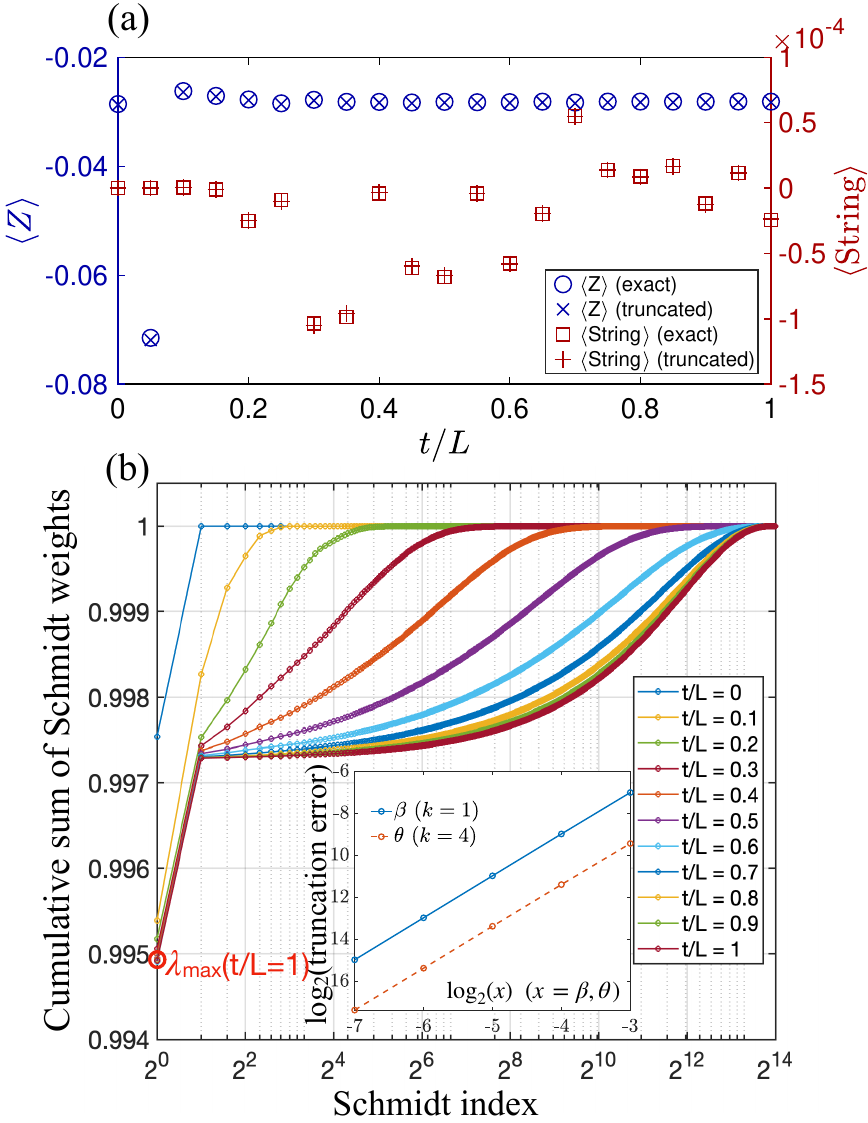}
\caption{(a) Comparison of the exact expectation values with those obtained from a truncated density matrix where we have kept \textit{only} the leading Schmidt vector of the canonical purification of a Gibbs state subjected to a local quantum quench. Here $H = \sum_{i=1}^{L-1} Z_i Z_{i+1}+ \sum_{i=1}^{L} (g X_i + h Z_i) + Z_1/4 - Z_L/4$ with $L = 14, g = 1.1, h = 0.35$. The specific parameters in the Hamiltonian are taken from Ref.~\cite{RodriguezNieva2024quantifying}. The local quench operator is $U_\theta = e^{-i \theta Z_{L/2}}$ with $\theta = 0.5$ and inverse temperature $\beta = 0.1$ (see Eq.\ref{eq:localquench}). (b) The cumulative sum of Schmidt weights as a function of the Schmidt index for various times. The left most point of the curve corresponds to the largest Schmidt eigenvalue $\lambda_{\textrm{max}}$ (e.g., $\lambda_{\textrm{max}}$ at $t = L$ is shown as circled), and the truncation error is the deviation of this sum from unity. The inset shows the truncation error as a function of $\beta$ at $\theta = 0.5$ while keeping single Schmidt state, and as a function of $\theta$ at $\beta = 0.8$ while keeping four Schmidt states. The slope is two on a log-log plot.}
\label{fig:z_string}
\end{figure}

We begin with a surprising observation that motivated this work. Consider subjecting a Gibbs state $\rho_\beta = e^{-\beta H}/Z$ to a local quench of the form 

\be \rho_\beta \to \rho_{\beta,\theta}(t) = e^{i Ht} U^{\dagger}_\theta(x) \rho_{\beta} U_\theta(x) e^{-i Ht} \label{eq:localquench}\ee 

where $U_\theta(x) = e^{i \theta O_1(x)}$ is a local unitary for some Hermitian operator $O_1(x)$ at site $x$ \footnote{As $\beta \to \infty$, our local quench protocol reduces to the ground-state local quench considered in Refs.~\cite{Nozaki2014quantum,nozaki2014notes,nozaki2013holographic,asplund2015holographic,caputa2014entanglement,bianchi2025entropy}. In this work, we are concerned only with $\beta \neq \infty$. Note also that our quench protocol, as well those in Refs.~\cite{Nozaki2014quantum,nozaki2014notes,nozaki2013holographic,asplund2015holographic,caputa2014entanglement,bianchi2025entropy}, is distinct from the `joining quench' considered in Ref.\cite{calabrese2007entanglement}}. At small $\theta$, the expectation value of any operator $O_2$ takes the form $\tr(\rho_{\beta,\theta}(t)O_2) = c + i \theta \tr\left( \rho_\beta [O_1(x,t),O_2]\right) + O(\theta^2)$, and analogously, at small $\beta$,  $\tr(\rho_{\beta,\theta}(t) O_2) = c' - \beta \tr\left(O'_1(x,t) O_2\right)/|\mathcal{H}| + O(\beta^2)$, where $c, c'$ are time-independent constants, $O'_1(x) = U^{\dagger}_\theta(x) H U_\theta(x) - H$ is a local operator, and $|\mathcal{H}|$ is Hilbert-space dimension.  Thus, the expectation values linear in $\theta$ or $\beta$ encode unequal-time correlation functions, as in the standard linear-response theory~\cite{kubo2012statistical}. It is convenient to express expectation values in terms of the canonical/standard purification~\cite{takahashi1975higher,schmutz1978real,terhal2002entanglement,dutta2021canonical}, $\tr(\rho_{\beta,\theta}(t) O) = \langle \sqrt{\rho_{\beta,\theta}(t)}|O|\sqrt{\rho_{\beta,\theta}(t)}\rangle$,
where  
\be 
|\sqrt{\rho_{\beta,\theta}(t)}\rangle = \left(\sqrt{\rho_{\beta,\theta}(t)} \otimes \mathds{1}_a\right) |\Phi\rangle \label{eq:canonicalpurification}
\ee 
with $|\Phi\rangle = \sum_i |i\rangle_s |i\rangle_a $ the unnormalized maximally entangled state between the system and an isomorphic ancilla. Note that for our problem, $\sqrt{\rho_{\beta,\theta}(t)} \propto \rho_{\beta/2,\theta}(t)$. This formulation allows a (not necessarily optimal) separation of classical and quantum correlations: if the canonical purification is short-range entangled (SRE), $\rho \otimes \mathds{1}_a$ can be written as a convex sum of SREs~\cite{chen2023symmetry}.
At least in $1+1$-D, the unperturbed state $|\sqrt{\rho_\beta}\rangle$
itself is expected to admit such an SRE decomposition, as suggested by area-law entanglement under the bipartition $A_sA_a|B_sB_a$ that groups each subsystem together with its ancilla~\cite{barthel2017one}, and by the
tensor-network approximability of the Gibbs state \cite{kuwahara2021improved}, \footnote{Viewed as a vector in the doubled Hilbert space, the entanglement of
$|\sqrt{\rho_\beta}\rangle$ across the cut $A_sA_a|B_sB_a$ is precisely the operator-space entanglement entropy (OSEE) of $\sqrt{\rho_\beta}$; the same identification applies to the time-evolved
state $|\sqrt{\rho_{\beta,\theta}(t)}\rangle$
\cite{zanardi2001entanglement,prosen2007operator}}. This raises the question: how entangled can $|\sqrt{\rho_{\beta,\theta}(t)}\rangle$ become after the quench---e.g., can Rényi entropies show volume-law scaling? Relatedly, are time-dependent observables captured by a few Schmidt modes? For non-integrable systems, one might expect even linear response to involve exponentially many Schmidt states.

As a concrete example, consider the Gibbs state of the mixed-field Ising model, $H = \sum_{i=1}^{L-1} Z_i Z_{i+1}+ \sum_{i=1}^{L} (g X_i + h Z_i)$, at $\beta \ll 1$, subjected to the above quench at $\theta = O(1)$. We compare the exact expectation value $\tr(\rho(t) O)$ with $\tr(\rho_{\textrm{trunc},A}(t) O)$ where $\rho_{\textrm{trunc},A}$ is obtained by performing a Schmidt decomposition
of $|\sqrt{\rho(t)}\rangle$ across a half-system bipartition $A_s A_a \mid B_s B_a$, and retaining only $k = O(1)$  largest Schmidt components. Fig.\ref{fig:z_string}(a) shows results for two operators: a single-site operator $O = Z_1$ and  a string operator spanning region A, $O = X_1Y_2Z_3X_4Y_5Z_6X_7$ for $L = 14$ and $k= 1$, i.e., we have kept  \textit{only} the leading Schmidt vector out of $2^{14}$. We find that the truncated expectation values track the exact ones very closely at all times, and correspondingly, the leading Schmidt state carries most of the weight as shown in Fig.\ref{fig:z_string}(b). Perhaps most interestingly, at small $\beta$ or small $\theta$, the truncation error after retaining $O(1)$ Schmidt states vanishes as $\theta^2$ or $\beta^2$ respectively (see the inset of Fig.\ref{fig:z_string}(b)). Therefore, the leading time-dependent signal, proportional to $\theta$ or $\beta$, is already encoded in an $O(1)$ number of top Schmidt states.

The rest of the paper explains and generalizes this observation. We first establish the mechanism in a minimal local-circuit setting which has a parameter $\epsilon$ analogous to $\beta, \theta$ above. As we will show, here the response linear in $\epsilon$ is encoded in a single Schmidt state. Relatedly, the Rényi entropies $S_{\alpha > 1}$ obey a constant law while $S_{\alpha \le 1}$ obey a volume law with the coefficient of $S_1$ parametrically small in $\epsilon$. Perhaps most interestingly, the Rényi entropies of the leading Schmidt state (defined via further bipartitioning) obey a constant law for $\alpha \ge \alpha_c$ and volume law for $\alpha <  \alpha_c$ where $\alpha_c < 1$. This implies~\cite{schuch2008entropy} that in $d=1$, the linear response for all times $t$ is encoded in a state that is approximable by a polynomial-bond-dimension matrix product state (MPS). We then return to local quenches in Gibbs states and show that this structure carries over closely.

\textbf{A circuit model.} The Hilbert space of our model consists of $N$ $q$-dimensional qudits on a $d$-dimensional lattice, with local basis $\{|0\rangle, |1\rangle, ..., |q-1\rangle\}$.  Consider a unitary circuit $U(t)$ consisting of geometrically local gates, and evolve a local, traceless, unitary operator $O$ as $O(t) = U(t) O U^{\dagger}(t)$. For now, we assume absence of any global symmetries (we will later briefly discuss a circuit with U(1) symmetry). The state of our interest is:

\begin{equation}
\ket{\psi(t)}
=
\frac{( \I + \eps O(t))\ket{0}}{\sqrt{\cN_t}}, \label{eq:randomcircuitstate}
\end{equation}
where $0 < \epsilon < 1$ is a real number, $\ket{0}$ is any product state, which will take to be the all-0 state ($ \ket{00...0}$), and $\cN_t = 1+\eps^2+2\eps \Re{a_t}$ is the normalization with
$a_t =\bra{0}O(t)\ket{0}$. Similar to the aforementioned discussion, for any operator $O'$, $\bra{\psi(t)}O' \ket{\psi(t)}$ encodes unequal-time connected correlator $\bra{0} \left(O^{\dagger}(t) O' + O' O(t)\right) \ket{0} - \bra{0}O^{\dagger}(t)\ket{0} \bra{0} O' \ket{0} - \bra{0} O' \ket{0} \bra{0}O(t)\ket{0}$ at linear order in $\epsilon$.

We now bipartition the system into two contiguous halves $A, B$ and study $\rho_A = \tr_B \ket{\psi(t)}\!\bra{\psi(t)}$. Decomposing
\(O(t)\ket{0}\) as
$O(t)\ket{0} = \ket{x}_A\otimes \ket{0_B}+
\ket{y}_{AB}$ where $\braket{0_B}{y}=0$, one finds the following convenient form for $\rho_A$:

\begin{equation}
\rho_A 
=
(1-\mu)\ket{v}\!\bra{v}+\mu\,\omega,
\label{eq:spike-cloud}
\end{equation}

where $\mu = \epsilon^2 \braket{y}{y}/\cN_t$, $\ket{v} = \frac{\ket{0_A}+\eps \ket{x}}{\sqrt{\cN_t(1-\mu)}}$ is a unit norm vector, $\omega = \tr_B \ket{y}\!\bra{y}/\braket{y}{y} $ is a normalized density matrix (note that since $O(t)\ket{0}$ is unit norm, $\braket{x}{x} + \braket{y}{y} = 1$). Eq.\ref{eq:spike-cloud} has several consequences. First, for correlation functions accurate only to $O(\epsilon)$ the $\mu \omega$ term in Eq.\ref{eq:spike-cloud} can be dropped, and one obtains a rank-1 approximant to $\rho_A$ \textit{for all time $t$}. Second, the decomposition in Eq.\ref{eq:spike-cloud} is ripe for obtaining sharp bounds on Rényi entropies $S_{\alpha}(\rho_A) =  \frac{1}{(1-\alpha)} \log(\tr \rho^\alpha_A)$. A short summary is as follows, see Appendix \ref{sec:renyi_inequalities} for details: 

(a) For $\alpha > 1$, $S_{\alpha}(\rho_A) \leq \frac{\alpha}{\alpha-1}\log(1/(1-\epsilon^2))$. This implies that Rényi entropies for $\alpha > 1$ do not scale with the system. This constant-law scaling has a simple origin: the decomposition in Eq.~\eqref{eq:spike-cloud} guarantees that the largest Schmidt eigenvalue satisfies \(\lambda_{\max}(\rho_A)\ge 1-\epsilon^2\) at all times.

(b) $\mu S_1(\omega) \le S_1(\rho_A) \le -\mu\log\mu -(1-\mu)\log(1-\mu)+ \mu S_1(\omega)$. This implies that the volume-law coefficient for $S_1$ is suppressed relative to that of $\omega$ by a factor of $\mu$ ($\approx \epsilon^2$ at small $\epsilon$). 

(c) For $\alpha < 1$, $S_\alpha(\omega)+\frac{\alpha}{1-\alpha}\log\mu
 \le
 S_{\alpha}(\rho_A)
 \le
 \frac{1}{1-\alpha}\log\Big[(1-\mu)^\alpha+\mu^\alpha e^{(1-\alpha)S_\alpha(\omega)}\Big]$, i.e., the volume-law coefficient of $S_{\alpha < 1}$ is identical to that for $\omega$. 

To make further progress, we make two assumptions. First, if $S_{\alpha}(\omega)$ is extensive for some $\alpha > 1$, then the state $\ket{v}$ in Eq.~\eqref{eq:spike-cloud}
is exponentially close to the leading Schmidt vector of $\rho_A$ and the corresponding leading Schmidt eigenvalue is separated from the rest of the spectrum by an $O(1)$ gap; see Appendix~\ref{sec:retained-vs-schmidt}
for a derivation.  An example where extensivity of $S_\alpha(\omega)$ can be established is circuits consisting of local Haar random gates where, at long times,  $\omega$ can effectively be replaced by the identity matrix within the intersection of the light cone with region $A$, and a product state outside it ~\cite{nahum2017quantum,nahum2018operator,keyserlingk2018operator,brandao2016local,harrow2023approximate}, \footnote{Note that if instead $S_\alpha(\omega)$ were area-law for some $\alpha<1$, then the aforementioned inequalities would imply an area-law bound for $S_\alpha(\rho_A)$ itself.}. By the above inequalities, the extensivity of $S_\alpha(\omega)$ also implies that $S_1(\rho_A)$ is extensive with a volume-law coefficient suppressed by $\epsilon^2$ at small $\epsilon$, while $S_{\alpha<1}(\rho_A)$ is extensive with an $O(1)$ volume-law coefficient inherited from $S_{\alpha < 1}(\omega)$. The physical origin of this R\'enyi-index-tuned transition is that the entanglement spectrum has a single $O(1)$ eigenvalue, approximately $1-\mu$, together with an exponentially large number of much smaller eigenvalues carrying total weight $\mu$ (Eq.\ref{eq:spike-cloud}). Further, assuming that at long times the state $O(t)\ket{0}$
has exponentially small overlap with any fixed product state, which is a reasonable assumption due to operator scrambling, one expects 
$|a_t|=|\bra{0}O(t)\ket{0}|=o(1)$ and $\langle x|x\rangle =o(1)$, or equivalently $\langle y|y\rangle \to 1$.
Then at long times $\mu\to \mu_\infty=\epsilon^2/(1+\epsilon^2)$.

\begin{figure}[t]
\centering
\includegraphics[width=\columnwidth]{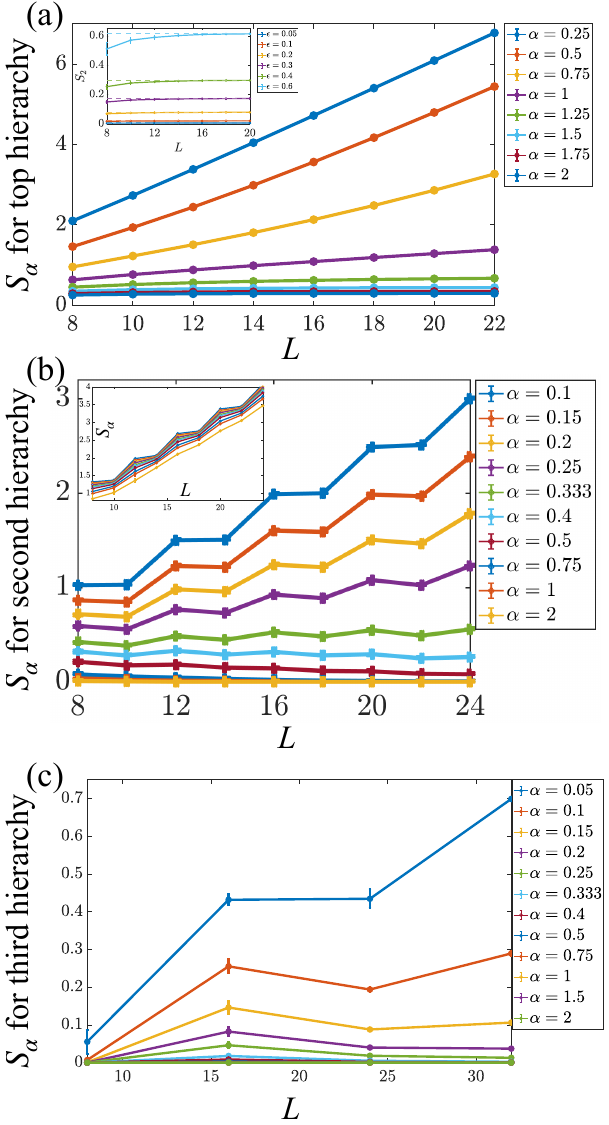}
\caption{Saturated Rényi entropies of the top Schmidt state corresponding to the first three hierarchies in the random circuit model (Eq.\ref{eq:randomcircuitstate}), with $L_A = L/2, L_{A_{1/4}} = \lfloor L/4 \rfloor, L_{A_{1/8}} = L/8$, at $\epsilon = 0.4$.  The plotted entropies are averaged over ten samples for (a) and four each for (b), (c). In (a), the error bars are smaller than the marker size. The inset of (a) shows $S_2(L)$ for several $\epsilon$, with dashed lines denoting the exact upper bound. The inset of (b) shows the Rényi entropies of the second highest Schmidt state in the second hierarchy. }
\label{fig:renyi_circuit_combo}
\end{figure}

The state $|v\rangle$ in Eq.\ref{eq:spike-cloud} is central to our discussion: it captures the $O(\epsilon)$ correlations and, under our assumptions, is the leading Schmidt state. Its entanglement structure therefore determines whether the linear response has efficient encoding. For example, in $d=1$, if $|v\rangle$ is approximable by a polynomial-bond-dimension MPS, then $|v\rangle \langle v|$ provides an area-law approximant for $\rho_A$. Following Ref.\cite{schuch2008entropy}, the key question is whether $S_\alpha(|v\rangle)$ is area-law for some $\alpha < 1$, uniformly over all cuts and times. To study this, we bipartition $A \equiv A_{1/2}$ into equal contiguous regions $A_{1/4}, B_{1/4}$. Notably,  $\ket{v} \propto \ket{0_A} + \epsilon \ket{x_{A_{1/2}}}$ has the
same structure as the state $\ket{\psi(t)}$, and its reduced density matrix also takes the analogous form,
\be 
\rho_{A_{1/4}}
=
(1-\mu_{1/4})\ket{v_{1/4}}\!\bra{v_{1/4}}
+
\mu_{1/4}\,\omega_{1/4}. \label{eq:rho_Aquarter}
\ee
The key difference between $\rho_{A_{1/4}}$ and $\rho_A \equiv \rho_{A_{1/2}}$ (Eq.\ref{eq:spike-cloud}) is that while $\mu$ remains an $O(1)$ constant for all times, $\mu_{1/4}$ is exponentially small in system size at long times. This is because when $t \sim L$, operator spreading makes the weight of $O(t)\ket{0}$ in any fixed
product-state slice of $B$, such as $\ket{0_B}$, exponentially small, which implies that $\mu_{1/4} \lesssim e^{-\gamma |A_{1/2}|}$ for some $\gamma > 0$, see Appendix~\ref{sec:entanglement_of_v}. Assuming that $\omega_{1/4}$ has extensive R\'enyi entropies,
$S_\alpha(\omega_{1/4}) \sim s_\alpha^{(1/4)} |A_{1/4}|$ for all $\alpha$, then $S_\alpha(\rho_{A_{1/4}})$ obeys a constant law for $\alpha \ge \alpha_c$, where
$\alpha_c = \frac{s_{\alpha_c}^{(1/4)}}{2\gamma + s_{\alpha_c}^{(1/4)}} < 1$. In fact, the constant value of $S_{\alpha > \alpha_c}$ approaches \textit{zero} in the thermodynamic limit. Thus, unlike the full state which undergoes a
constant-law to volume-law transition at $\alpha_c = 1$, the dominant Schmidt state $\ket{v}$ is expected to undergo an analogous transition at $\alpha_c < 1$. This conclusion holds for all extensive bipartitions, while cuts involving only $O(1)$ degrees of freedom are trivially
bounded.  See Appendix~\ref{sec:entanglement_of_v} for details. We thus conclude that \(\ket{v}\) is approximable by a polynomial-bond-dimension MPS in \(d=1\).

The above construction leads to an infinite nested hierarchy. Subdividing $A_{1/4}$ into $A_{1/8}, B_{1/8}$, and iterating, gives a critical index $\alpha_{c,j}$ at level $j$ of hierarchy: the density matrix $\rho_{A_{2^{-j}}} = \tr_{B_{2^{-j}}} |v_{2^{-j+1}}\rangle \langle v_{2^{-j+1}}|$, corresponding to the top Schmidt state $|v_{2^{-j+1}}\rangle$ of $\rho_{A_{2^{-j+1}}}$, has a constant Rényi entropy $S_\alpha$ for $\alpha > \alpha_{c,j}$. Repeating the aforementioned argument, one finds $\alpha_{c,j}$ decreases with $j$. In a solvable maximally scrambled limit, where  $\omega_{2^{-j}}$ is maximally mixed on $N/2^j$ qudits, one finds $\alpha_{c,j} = 1/(2^j - 1)$, see Appendix \ref{sec:scrambled-hierarchy}.

\begin{figure}[t]
\centering
\includegraphics[width=0.8\columnwidth]{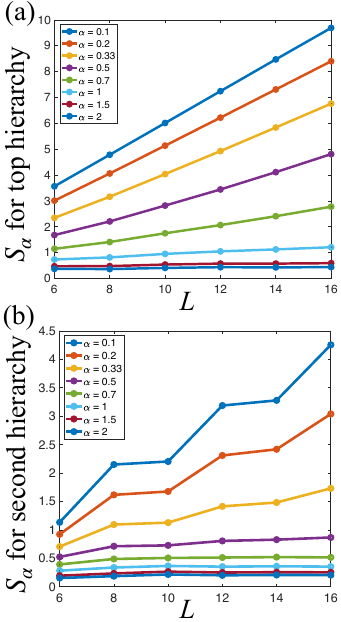}
\caption{(a) Saturated bipartite Rényi entropies of the state $|\sqrt{\rho_{\beta,\theta}(t)}\rangle$ for the mixed-field Ising model at $g = 1.1, h = 0.35, \beta = 1, \theta = 0.5, L_A = L/2$. (b) Saturated bipartite Rényi entropies for the leading Schmidt state of $|\sqrt{\rho_{\beta,\theta}(t)}\rangle$ with $L_{A_{1/4}} = \lfloor L_A/2 \rfloor$ defined via further bipartitioning region $A$ into two equal halves.}
\label{fig:combo_gibbs}
\end{figure}

\underline{Numerical study:} We test these predictions by ED for a one-dimensional chain of \(L\) qubits with open boundaries. The state is Eq.\ref{eq:randomcircuitstate} with $U(t)$ a depth-$t$ brickwork circuit of independent Haar-random nearest-neighbor two-qubit gates, $O = Z_{L/2}$, and $t \sim r L$ with $r = O(1)$. Figs.\ref{fig:renyi_circuit_combo}(a), (b), (c) show the saturated Rényi entropies of $\rho_{A_{1/2^{j}}}$ for the leading Schmidt eigenvector across three consecutive hierarchies $j = 1, 2, 3$ at $\epsilon = 0.4$. The entropies grow monotonically before saturation (see Appendix \ref{sec:timedepend}). We find clear evidence for constant-to-volume-law transition for all three hierarchies with critical indices close to the aforementioned solvable-limit values $\alpha_{c,j} = 1/(2^{j}-1) = 1, 1/3, 1/7$ for $j=1,2,3$ respectively \footnote{For \(j=2,3\), the dominant Schmidt vector(s) across the half-chain cut were obtained without forming or fully diagonalizing
\(\rho_{A_{1/2}}\). Instead, writing the reshaped wavefunction as \(\Psi\), we used standard power/subspace iteration with the matrix-free action
\[
x \mapsto \rho_{A_{1/2}}x = \Psi(\Psi^\dagger x),
\]
with convergence checked by the eigenpair residual; see, e.g., Saad~\cite{Saad2011}.}. The inset of Fig.\ref{fig:renyi_circuit_combo}(a) shows that saturated $S_2$ is very close to the aforementioned exact bound $S_2 \leq 2 \log(1/(1-\mu_\infty))$. The inset of   Fig.\ref{fig:renyi_circuit_combo}(b) shows $S_\alpha$ for the \textit{second}-highest Schmidt state of $\rho_{A_{1/2}}$, defined via the bipartition $A_{1/4}|B_{1/4}$. It is strikingly different from that for the leading Schmidt state (main panel of Fig.\ref{fig:renyi_circuit_combo}(b)), and exhibits a maximal-slope volume law for all $\alpha$, consistent with Eq.\ref{eq:spike-cloud}: the second Schmidt state belongs to the scrambled density matrix $\omega$, which is close to identity.

To confirm that our qualitative results are not tied to the Haar-random architecture, we also studied two other architectures: (i) a nearest-neighbor Clifford+\(T\) circuit drawn from a finite universal gate set; see Appendix~\ref{sec:universal-gateset-circuit}, and (ii) a U(1)-symmetric circuit~\cite{khemani2018operator,Rakovszky2018diffusive}; see Appendix~\ref{sec:U1circuit}. In both cases we find $\alpha_c = 1$ for $\rho_{A_{1/2}}$ and $\alpha_c < 1$ for $\rho_{A_{1/4}}$. We leave an analytical study of symmetric circuits for future work.

\textbf{Local quenches on Gibbs states.} We now return to local quenches on Gibbs states (recall Fig.~\ref{fig:z_string} and the associated discussion) and demonstrate the close parallel between the canonical purification $|\sqrt{\rho_{\beta,\theta}(t)}\rangle$ and the circuit-model state $|\psi(t)\rangle$ (Eq.~\eqref{eq:randomcircuitstate}).

The unperturbed purification $|\sqrt{\rho_\beta}\rangle$ is expected to obey an area law for bipartite Rényi entropy for all indices $\alpha$, in particular, $S_{1/2} \le c L_A^{d-1}$, where $c = O(1)$~\cite{barthel2017one}. One readily verifies that $|\sqrt{\rho_{\beta,\theta}(t)}\rangle$ has a non-zero $O(1)$ overlap with  $|\sqrt{\rho_\beta}\rangle$ for all time. Using techniques similar to Ref.\cite{muth2011dynamical}, this has two remarkable consequences, both paralleling the circuit model. First, $S_{\alpha > 1}$ obeys an area law at all times: \be S_{\alpha > 1} \le 
\left(c L^{d-1}_A + \frac{\alpha}{\alpha-1} \beta \normop{\Delta H}
\right),\ee where $\Delta H = U_\theta^{\dagger}(x) H U_\theta(x) - H$ is local and hence has a bounded norm in a spin system. Second, the largest Schmidt eigenvalue of $\rho_A$ ($A = A_s A_a$) is $O(1)$ for all times; see Appendix \ref{sec:renyisGibbs}.
The area law for $S_{\alpha > 1}$ is reminiscent of numerical results in Ref.\cite{alba2025more}, where  the Rényi OSEE of operators with nonzero trace in certain non-integrable spin-chains was found to  grow only logarithmically with time at indices $\alpha > 1$. In our case, $\sqrt{\rho_{\beta,\theta}(t)}$ also has nonzero trace, but its persistent overlap with $\sqrt{\rho_\beta}$ leads to an even stronger suppression: an area law for $\alpha > 1$.

We numerically study the entanglement structure of $|\sqrt{\rho_{\beta,\theta}(t)}\rangle$ for the mixed-field Ising model, $H = \sum_{i=1}^{L-1} Z_i Z_{i+1}+ \sum_{i=1}^{L} (g X_i + h Z_i) + Z_1/4 - Z_L/4$ with the local quench unitary $U_\theta = e^{-i \theta Z_{L/2}}$. The Hamiltonian parameters are taken from Ref.\cite{RodriguezNieva2024quantifying} to maximize certain chaos diagnostics \footnote{The boundary terms ensure that the system does not have any spatial symmetries.}. As in the circuit model, $S_\alpha$ shows clear numerical evidence of an area-law to volume-law transition at $\alpha = 1$. The mechanism is again the entanglement spectrum of $\rho_A$, which contains an $O(1)$ number of large eigenvalues, as mentioned above on general grounds and also visible in Fig.\ref{fig:z_string}(b). The entanglement structure of these dominant Schmidt states also parallels that in the circuit model: their Rényi entropies, defined via bipartitioning $A$, undergo an area-law to volume-law transition at $\alpha_c < 1$, see Fig.~\ref{fig:combo_gibbs}(b). Therefore, we expect that the dominant Schmidt states are again approximable by polynomial-bond-dimension MPS. All entropies essentially grow monotonically with time before saturating; see Appendix \ref{sec:timedepend}.

The parallel between the circuit model and the Gibbs state becomes sharper at small $\beta$ or $\theta$. For $\beta \ll 1$ at fixed $\theta = O(1)$,
one finds $\rho_A(t)
= |v(t)\rangle\langle v(t)|+ O(\beta^2)$ where $|v(t)\rangle = |\phi_A(t)\rangle -
\frac{\beta}{2}|x(t)\rangle_A
+
O(\beta^2)$ and $|\phi_A(t)\rangle, |x(t)\rangle$ are some vectors defined analogously to the circuit model. Thus, to linear order in $\beta$, the dynamics is encoded in a
single state on $A$.
Similarly, for $\theta \ll 1$ at fixed $\beta = O(1)$,
$\rho_A(t)
=
\sum_{i=1}^{D}
\Bigl(
\sqrt{p_i}\,|\Phi_A^i\rangle
+
i\theta\,|x_A^i(t)\rangle
\Bigr)
\Bigl(
\sqrt{p_i}\,\langle \Phi_A^i|
-
i\theta\,\langle x_A^i(t)|
\Bigr)
+
O(\theta^2)$,
where $D$ is the number of Schmidt states retained in a
finite-dimensional approximation to the thermal purification
$|\sqrt{\rho_\beta}\rangle$, and the vectors $|\Phi_A^i\rangle, |x_A^i(t)\rangle$ are defined analogously. See Appendices \ref{sec:smallbeta},\ref{sec:smalltheta}. In both limits, the leading time-dependent signal is
therefore captured by \(O(1)\) Schmidt states, with the corresponding truncation error vanishing as \(\beta^2\) (\(\theta^2\)) at small \(\beta\) (\(\theta\)), as also shown in Fig.~\ref{fig:z_string}(b). To be explicit, this implies that for all times \(t\), there exists a collection of \(D=O(1)\) (unnormalized) MPS-approximable states \(\{|v_i(t)\rangle\}_{i=1}^D\subset \mathcal H_A\) such that, for any \(O_2\) supported in \(A\) with \(\|O_2\|=O(1)\),
\begin{eqnarray}
& & \hspace{-1cm}\sum_{i=1}^{D}\langle v_i(t)|O_2|v_i(t)\rangle =  \nonumber \\
& & \hspace{-1cm}\tr(\rho_{\beta}O_2) + 
i\theta\,\tr\!\left(\rho_\beta[O_1(x,t),O_2]\right)
+ O(\theta^2), \nonumber
\end{eqnarray}
with the \(O(\theta^2)\) term uniformly bounded in time. An analogous statement holds in the small-\(\beta\) expansion. The above results when combined with Fannes--Audenaert inequality~\cite{fannes1973,audenaert2007sharp} imply that the volume-law coefficient of the von Neumann entropy, $\lim_{L\to\infty} \frac{S_1(\rho_A)}{L}$,
is at most of order $\beta^2$ ($\theta^2$) at small $\beta$ ($\theta$). This is consistent with the numerical results in Appendix~\ref{sec:volume-law-coeff}, which show that the volume-law coefficient of $S_1$ vanishes approximately as $\beta^2$ ($\theta^{1.8}$) at small $\beta$ ($\theta$). The small deviation from quadratic scaling in $\theta$ is likely a finite-size effect. We also studied the regime $\beta, \theta \ll 1$ and found that the truncation error after keeping $O(1)$ Schmidt states vanishes as $(\beta\theta)^2$, as expected; see Appendix~\ref{sec:beta_theta_small}.

\textbf{Discussion.} We expect that the phenomena described here should extend to Gibbs states of generic interacting  quantum field theories, including holographic ones.  Although related local-quench and purification setups have been studied ~\cite{Nozaki2014quantum,nozaki2014notes,nozaki2013holographic,asplund2015holographic,caputa2014entanglement,dutta2021canonical,kusuki2021dynamics,kudler2021quasi,bianchi2025entropy}, we are not aware of a treatment of the locally quenched canonical purification in Eq.~\eqref{eq:canonicalpurification} for a strongly interacting field theory. Such a calculation may be subtle. Even for local quenches on ground states, the analytic continuation from R\'enyi entropies to obtain $S_1$ can involve order-of-limits issues ~\cite{caputa2014entanglement}. In the present setting, this issue is especially sharp: although $S_\alpha$ is smooth as $\alpha\to1$ at any finite system size, our results imply that the limits $\alpha\to1$ and $L\to\infty$ do not commute, leading to an area law for $S_{\alpha>1}(\rho_{A_{1/2}})$ but a volume-law component in $S_{\alpha \le 1}(\rho_{A_{1/2}})$.

We emphasize that our results on area-law approximants are statements about the entanglement structure and representability---we showed the existence of such a representation by truncating the \textit{exact state} obtained using ED. This does not imply the existence of a general-purpose classical algorithm for finding the approximant with polynomial resources! This distinction is important: low bond-dimension by itself does not guarantee that one can \textit{find} such a representation efficiently. Indeed, there are one-dimensional Hamiltonians whose ground states are MPS, yet finding the relevant MPS can be computationally hard \cite{Schuch2008computational,Jiang2025local}. On that note, it is useful to contrast the top level of the hierarchy with the lower levels. If the full time-evolved state had an area law for \(S_\alpha\) at some \(\alpha<1\), uniformly over cuts and over poly($L$) times, then in one dimension it would admit a polynomial-bond-dimension MPS simulation~\cite{Vidal04,schuch2008entropy}. Such a conclusion would be highly unexpected for generic chaotic systems such as the ones considered here. Therefore, the top-level value \(\alpha_c=1\) is consistent with complexity-theoretic expectations. Similarly, with regard to our circuit model, BQP/PromiseBQP computations can be formulated as constant-gap estimates of local observables such as $\langle 0|O(t)|0\rangle$~\cite{Watrous2008quantum,Janzing2007simple,Shor1997polynomial,Aharonov2009polynomial}. Such observables can be embedded into the linear-response form of Eq.~\eqref{eq:randomcircuitstate} \footnote{For example, add a flag/ancilla qubit $a$ to the total system, take \(O=Z_x X_a\), and evolve everything except the flag qubit with the circuit $U$, i.e., $O(t) = U Z_x U^{\dagger} X_a$. Then, for the state $|\psi(t)\rangle$ in Eq.~\eqref{eq:randomcircuitstate}, the term proportional to $\epsilon$ in $\langle \psi(t)| X_a|\psi(t)\rangle$ is $2\langle0|U Z_x U^{\dagger}|0\rangle$.}. Thus, if for BQP-complete circuit families the corresponding state \(|v\rangle\) remains area law in the sense discussed above, and if one can also efficiently track it for poly($L$) times with sufficient accuracy to resolve the \(O(\epsilon)\) signal, then one would obtain \(\mathrm{BQP}\subseteq\mathrm{BPP}\). This is not expected, which again underlines the difficulty of finding such a general-purpose classical algorithm. Nonetheless, there is a growing and substantial literature on classical simulation of quantum dynamics in one dimension, where the emphasis is typically on approximate, or in some regimes controlled, tensor-network and operator-space methods for selected observables or dynamical regimes~\cite{leviatan2017quantum,Houschild2018finding,White2018quantumdynamics,Rakovszky2022dissipation,Keyserlingk2022Operator,YiThomas2024Comparing,Artiaco2024Efficient,Angrisani2025classically,rudolph2025pauli,cruz2025quantum,anand2026kpz}. It would therefore be interesting to explore new algorithms that exploit the structure discussed in our work.

\textit{\underline{Acknowledgments:}} I  thank John McGreevy, Tadashi Takayanagi and Yu-Hsueh Chen for helpful comments on the manuscript. I  acknowledge use of ChatGPT (OpenAI) for help with writing exact-diagonalization codes. This work is supported by
the National Science Foundation under Grant No. DMR-2521369.

%

\appendix
\onecolumngrid

\newpage

\section{Details of inequalities for Rényi entropies} \label{sec:renyi_inequalities}

Our starting point is Eq.\ref{eq:spike-cloud} in the main text:

\begin{equation}
\rho_A 
=
(1-\mu)\ket{v}\!\bra{v}+\mu\,\omega,
\label{eq:spike-cloud-app}
\end{equation}

where $\mu = \epsilon^2 \braket{y}{y}/\cN_t$, $\ket{v} = \frac{\ket{0_A}+\eps \ket{x}}{\sqrt{\cN_t(1-\mu)}}$ is a unit norm vector, $\omega = \tr_B \ket{y}\!\bra{y}/\braket{y}{y} $ is a normalized density matrix (with the constraint $\braket{x}{x} + \braket{y}{y} = 1$). We are interested in point-wise (i.e. for a single realization of the circuit $U(t)$), as well as ensemble-averaged bounds on the Rényi entropies $S_\alpha = \frac{1}{1-\alpha} \log(\tr(\rho^\alpha_A))$. Let's write the eigendecomposition of $\rho_A$ as $\rho_A = \sum_i \lambda_i \ket{\lambda_i}\!\bra{\lambda_i}$ with $\lambda_1 \ge \lambda_2 \ge \lambda_3 \dots$. We consider the three cases of interest in turn: (a) $\alpha > 1$, (b) $\alpha = 1$, and (c) $\alpha < 1$. 

\begin{enumerate}[(a)]
\item $S_{\alpha > 1}$: First, we notice that $\lambda_1 \ge \bra{v} \rho_A \ket{v} \ge (1-\mu)$. Therefore, $\tr(\rho^\alpha_A) = \sum_i \lambda^\alpha_i \ge \lambda^\alpha_1 \ge (1-\mu)^\alpha$, and therefore, one obtains $S_{\alpha > 1} \le \frac{\alpha}{\alpha-1}\log(1/(1-\mu))$. Further, one can show that $\mu \le \epsilon^2$ (see next paragraph), which implies that 
\begin{equation}
S_{\alpha > 1} \le \frac{\alpha}{\alpha-1}\log(1/(1-\epsilon^2)).
\end{equation}
Note that this bound holds point-wise and therefore is also true for $\mathbb{E} \,S_{\alpha > 1}$, the ensemble averaged Rényi entropies.

It remains to prove \(\mu\le \eps^2\). Since \(\|x\|^2 := \braket{x}{x} =1-\|y\|^2\),
\begin{align}
\cN_t
&=
\|\ket{0_A}+\eps \ket{x}\|^2+\eps^2 \|y\|^2 \nonumber \\
&\ge
(1-\eps\|x\|)^2+\eps^2 \|y\|^2 \nonumber \\
&=
(1-\eps\sqrt{1-\|y\|^2})^2+\eps^2 \|y\|^2 \nonumber \\
&=
\|y\|^2+(\sqrt{1-\|y\|^2}-\eps)^2 \nonumber \\
&\ge \|y\|^2. \nonumber
\end{align}
Hence $\mu = \frac{\eps^2 \|y\|^2}{\cN_t}\le \eps^2$.

\item $S_1$: We apply the standard Holevo-type inequality $\sum_i p_i \log(1/p_i) + \sum_i p_i S(\rho_i) \ge S(\sum_i p_i \rho_i) \ge \sum_i p_i S(\rho_i)$ to Eq.\ref{eq:spike-cloud-app}:

\begin{equation}
\mu S(\omega)
\le
S(\rho_A)
\le
h_2(\mu)+\mu S(\omega),
\label{eq:vN-pointwise}
\end{equation}
where \(h_2(x):=-x\log x-(1-x)\log(1-x)\).
Taking the ensemble average of these pointwise inequalities, one obtains, $\mathbb{E}\, [\mu S(\omega)]
\le
\mathbb{E}\, [S(\rho_A)]
\le
\mathbb{E}\, [h_2(\mu)]+\mathbb{E}\, [\mu S(\omega)]$.

These inequalities imply that the volume law coefficient of $S(\rho_A)$ is given by $\lim_{N \to \infty} \mu S(\omega)/N$, and is therefore upper bounded by $\epsilon^2 \log(q)$.

\item $S_{\alpha < 1}$:  We now show that when $\alpha < 1$,

\begin{equation}
\mu^\alpha\tr\omega^\alpha
\le
\tr \rho^\alpha_A
\le
(1-\mu)^\alpha+\mu^\alpha\tr\omega^\alpha.
\end{equation}
Equivalently,
\begin{equation}
S_\alpha(\omega)+\frac{\alpha}{1-\alpha}\log\mu
\le
S_\alpha(\rho_A(t))
\le
\frac{1}{1-\alpha}\log\Big[(1-\mu)^\alpha+\mu^\alpha e^{(1-\alpha)S_\alpha(\omega)}\Big].
\label{eq:n>1pointwise}
\end{equation}

\begin{proof}
For \(0<\alpha<1\), the map \(x\mapsto x^\alpha\) is operator monotone increasing. Since \(\rho_A\ge \mu \omega\), monotonicity gives
\begin{equation}
\rho_A^\alpha\ge (\mu\omega)^\alpha=\mu^\alpha\omega^\alpha.
\end{equation}
Taking the trace gives the lower bound in Eq.\ref{eq:n>1pointwise}. The upper bound follows from the Rotfel'd inequality~\cite{bhatia2013matrix}:
\begin{equation}
\tr(A+B)^\alpha\le \tr A^\alpha+\tr B^\alpha
\qquad (0<\alpha<1),
\end{equation}
applied to \(A=(1-\mu)\ket{v}\!\bra{v} \) and \(B=\mu\omega\).
\end{proof}

\end{enumerate}

\section{Bound on the overlap between $|v\rangle$ and the top Schmidt state of $\rho$}\label{sec:retained-vs-schmidt}

We again start with Eq.\ref{eq:spike-cloud-app}.
Again writing $\rho_A = \sum_i \lambda_i |\lambda_i \rangle \langle \lambda_i|$ with $\lambda_1 \geq \lambda_2 \geq ...$, we wish to prove that $\ket{\lambda_1}$ has a large overlap with $\ket{v}$. The basic idea of the proof is to relate the overlap $\langle \lambda_1|v\rangle$ to the norm of the state $Q \rho_A |v\rangle$, where $Q = \mathds{1} - \ket{v}\!\bra{v}$ is the projection orthogonal to $|v\rangle$. It will be useful to define the following scalars, $a := \bra{v}\omega\ket{v}, \beta := \|\omega\|_\infty, m := \langle v |\rho_A|v\rangle = 1-\mu + \mu a$, and $\Delta := 1-\mu+\mu a-\mu \beta$. While $\Delta$ is unconditionally positive for $\mu < 1/2$, in physical applications, we expect both $a$ and $\beta$ to be exponentially small in the system size due to operator scrambling, and therefore, $\Delta$ is expected to be positive and an $O(1)$ number independent of system size for all $0 < \mu < 1$. We now state the main result.

\begin{proposition}
\label{prop:retained_vs_top}
If $\Delta > 0$, the top eigenvector $\ket{\lambda_1}$ of $\rho_A$  obeys the bound:
\begin{equation}
1-|\braket{\lambda_1}{v}|^2 \le \frac{\mu^2\bigl(\bra{v}\omega^2\ket{v}-a^2\bigr)}{\Delta^2} \le \frac{\mu^2 a}{\Delta^2}.
\end{equation}
\end{proposition}

\begin{proof}
    We first show that $m - \lambda_2 \geq \Delta$. The basic idea is to use Courant-Fischer min-max theorem. This theorem states that 

    \begin{equation}
\lambda_2
=
\min_{\substack{\textrm{any}\,\, \ket{\phi} \\ \|\phi\|=1}}
\ \max_{\substack{\ket{\psi} \perp \ket{\phi}\\ \|\psi\|=1}}
\bra{\psi}\rho_A\ket{\psi}.
\label{eq:CF}
\end{equation}

The optimal $\ket{\phi}$ in the whole Hilbert space is given by $\ket{\lambda_1}$ in which case this theorem is the standard variational principle for finding the first excited state. However, for us, it is expedient to instead use $\ket{\phi} = \ket{v}$. Using our decomposition in Eq.\ref{eq:spike-cloud-app} for $\rho$ then implies 
\be 
\lambda_2 \leq \max_{\substack{\ket{v} \perp \ket{\psi}\\ \|\psi\|=1}}
\mu \bra{\psi}\omega\ket{\psi} \leq \mu \|\omega\|_\infty = \mu \beta 
\ee

Therefore, 

\begin{equation}
m - \lambda_2 \ge m - \mu\beta = 1-\mu+\mu a-\mu\beta = \Delta > 0. \label{eq:lambda2bound}
\end{equation}

As hinted above, it will be useful to calculate the norm of the state $Q\rho_A\ket{v}$. A simple calculation gives

\begin{equation}
\|Q\rho_A\ket{v}\|^2 = \mu^2 \bra{v}\omega Q \omega \ket{v} = \mu^2\bigl(\bra{v}\omega^2\ket{v}-a^2\bigr).
\end{equation}
Because $\omega$ is a density matrix with eigenvalues $\le 1$, we have the operator inequality $\omega^2 \le \omega$. This implies $\bra{v}\omega^2\ket{v} \le \bra{v}\omega\ket{v} = a$, which provides an upper bound $\|Q\rho_A\ket{v}\|^2 \le \mu^2 a$.

Notice that using the definition $\bra{v} \rho_A \ket{v} = m$, $\|Q\rho_A\ket{v}\|^2 = \|(\rho_A-m)\ket{v}\|^2$. Now we expand $\ket{v} = \sum c_k \ket{\lambda_k}$ in the eigenbasis of $\rho_A$ which yields:
\bea
\|Q\rho_A\ket{v}\|^2 & = & \sum_{k \geq 1} |c_k|^2(\lambda_k-m)^2 \nonumber \\ 
& \geq & \sum_{k \geq 2} |c_k|^2 (\lambda_k-m)^2 \nonumber \\  
& \geq & \Delta^2 \sum_{k \geq 2} |c_k|^2 \nonumber \\
& = & \Delta^2 (1-|c_1|^2), \label{eq:boundQrhov}
\eea
where the second inequality in Eq.\ref{eq:boundQrhov} follows from Eq.\ref{eq:lambda2bound}.

Finally using $\|Q\rho_A\ket{v}\|^2 \le \mu^2 a$ and noticing that $c_1 = \langle \lambda_1|v\rangle$, the inequality in Eq.\ref{eq:boundQrhov} precisely yields our target bound:
\begin{equation}
1-|\braket{\lambda_1}{v}|^2 \le \frac{\|Q\rho_A\ket{v}\|^2}{\Delta^2} \le \frac{\mu^2 a}{\Delta^2}.
\end{equation}

\end{proof}

Proposition \ref{prop:retained_vs_top} is most useful when $a$ is small and $\Delta = \mathcal{O}(1)$. In the physical applications below, \(a\) and \(\beta=\|\omega\|_\infty\) are exponentially small, so for fixed \(\epsilon<1\), \(\Delta\) is \(O(1)\). The smallness of $a$ is captured by the following lemma.

\begin{lemma}
\label{lem:entropy_implies_backflow}
If the Rényi entropy $S_\alpha(\omega) \ge s_\alpha |A|$ for some index $\alpha>1$, then
\begin{equation}
a = \bra{v}\omega\ket{v} \le e^{-\frac{\alpha-1}{\alpha}s_\alpha |A|}.
\end{equation}
\end{lemma}

\begin{proof}
The intuition is that extensive entropy implies that the eigenvectors cannot be too concentrated along any single direction, including $\ket{v}$. Therefore, $a = \bra{v}\omega\ket{v}$ must be small. Defining $P = \ket{v}\!\bra{v}$, we use the Hölder inequality for Schatten norms:
\begin{equation}
a = \tr(\omega P) \le \|\omega\|_\alpha  \|P\|_{\alpha/(\alpha-1)} = \|\omega\|_\alpha,
\end{equation}
where we have used the fact that $P$ is a rank-one projector (and thus its Schatten norm is $1$ for any exponent). 

Using the definition of the Schatten norm $\|\omega\|_{\alpha} = [\tr(\omega^{\alpha})]^{1/\alpha}$, and recalling that $\tr(\omega^\alpha) = e^{(1-\alpha)S_\alpha(\omega)}$, this implies:
\begin{equation}
a \le e^{-\frac{\alpha-1}{\alpha}s_\alpha |A|}.
\end{equation}
\end{proof}

In a generic circuit, $S_\alpha(\omega)$ is expected to be extensive for all $\alpha$ due to operator scrambling.

\section{Entanglement structure of $\ket{v}$}\label{sec:entanglement_of_v}

To study the entanglement structure of the state $\ket{v}$, we focus on
the long-time regime $t \sim L$ and divide
$A \equiv A_{1/2}$ into two contiguous regions
$A_{1/4}, B_{1/4}$, with $|A_{1/4}| = |B_{1/4}| = N/4$ (the arguments below only rely on $|A_{1/4}|, |B_{1/4}|$ being extensive).
Since
\[
\ket{v} \propto \ket{0_A} + \epsilon \ket{x_{A_{1/2}}},
\]
it has the same structure as the state $\ket{\psi(t)}$ defined on the
full Hilbert space. This suggests the decomposition
\begin{equation}
\ket{x_{A_{1/2}}}
=
\ket{x_{A_{1/4}}}\otimes \ket{0_{B_{1/4}}}
+
\ket{y_{A_{1/2}}},
\qquad
\braket{0_{B_{1/4}}}{y_{A_{1/2}}} = 0.
\end{equation}

Tracing out $B_{1/4}$ yields
\begin{equation}
\rho_{A_{1/4}}
=
\tr_{B_{1/4}} \ket{v}\bra{v}
=
(1-\mu_{1/4})\ket{v_{1/4}}\!\bra{v_{1/4}}
+
\mu_{1/4}\,\omega_{1/4},
\end{equation}
where
\[
\ket{v_{1/4}} \propto \ket{0_{A_{1/4}}} + \epsilon \ket{x_{A_{1/4}}},
\]
\begin{equation}
\omega_{1/4}
=
\frac{
\tr_{B_{1/4}} \ket{y_{A_{1/2}}}\!\bra{y_{A_{1/2}}}
}{
\braket{y_{A_{1/2}}}{y_{A_{1/2}}}
},
\end{equation}
and
\begin{equation}
\mu_{1/4}
=
\frac{\epsilon^2 \|y_{A_{1/2}}\|^2}
{\|\ket{0_{A_{1/2}}} + \epsilon \ket{x_{A_{1/2}}}\|^2} \approx \epsilon^2 \|y_{A_{1/2}}\|^2.
\end{equation}

At long times one expects
\begin{equation}
\|x_{A_{1/2}}\|^2
=
\|\bra{0_B} O(t)\ket{0}\|^2
\sim e^{-\gamma |A_{1/2}|},
\end{equation}
since by our assumption, $O(t)\ket{0}$ carries exponentially small weight after projecting onto any product state in $B$, such as $|0_B\rangle$.  Orthogonality of the decomposition $\ket{x_{A_{1/2}}}
=
\ket{x_{A_{1/4}}}\otimes \ket{0_{B_{1/4}}}
+
\ket{y_{A_{1/2}}}$ implies $\|x_{A_{1/2}}\|^2
= \|x_{A_{1/4}}\|^2
+ \|y_{A_{1/2}}\|^2$, which then leads to
\begin{equation}
\mu_{1/4} \lesssim e^{-\gamma |A_{1/2}|}.
\end{equation}
We expect this bound to be close to saturated, since
$\ket{y_{A_{1/2}}}$ is again expected to carry most of the weight of
$\ket{x_{A_{1/2}}}$ at long times.

To determine the R\'enyi entropies, we assume that
$\omega_{1/4}$ has extensive entropy,
\begin{equation}
S_\alpha(\omega_{1/4})
\sim
s_\alpha^{(1/4)} |A_{1/4}|.
\end{equation}
Using Rotfel'd inequality for $0<\alpha<1$,
\begin{equation}
\tr(A+B)^\alpha \le \tr A^\alpha + \tr B^\alpha,
\end{equation}
we obtain
\begin{equation}
S_\alpha(\rho_{A_{1/4}})
\le
\frac{1}{1-\alpha}
\log\!\left[
(1-\mu_{1/4})^\alpha
+
\mu_{1/4}^\alpha \tr(\omega_{1/4}^\alpha)
\right].
\end{equation}

Substituting the scaling of $\mu_{1/4}$ and
$S_\alpha(\omega_{1/4})$, we find that
$S_\alpha(\rho_{A_{1/4}})$ obeys a constant law for
$\alpha \ge \alpha_c$, where
\begin{equation}
\alpha_c
=
\frac{s_{\alpha_c}^{(1/4)}}{2\gamma + s_{\alpha_c}^{(1/4)}} < 1.
\end{equation}

Note that at $\alpha = \alpha_c$, one obtains a constant law, unlike the top-level hierarchy where $\alpha_c = 1$ and $S_1(\rho_{A_{1/2}})$ is volume-law with a coefficient that scales as $\epsilon^2$ (see Appendix~\ref{sec:renyi_inequalities}). Further, for $\alpha > \alpha_c$, $S_{\alpha}$ \textit{vanishes} in the thermodynamic limit since the top Schmidt eigenvalue approaches unity asymptotically, while all others vanish. This is consistent with our numerical results, see Fig.\ref{fig:renyi_circuit_combo}(b).

The above argument applies to any bipartition where $|A_{1/4}|,|B_{1/4}|$ are extensive. If either of these two regions is $O(1)$ size, then one trivially obtains a bound that all entropies are $O(1)$. Therefore, the constant-law bound extends to all contiguous bipartitions of $A$,
which in one dimension implies efficient MPS compressibility of
$\ket{v}$~\cite{schuch2008entropy}.

\section{Critical Rényi index for nested hierarchy of Schmidt states for ``maximally scrambled'' circuits}\label{sec:scrambled-hierarchy}

In this appendix we consider a solvable limiting case in which the
state $\ket{\phi_t}:=O(t)\ket{0}$ behaves, under projection onto any
fixed product state on an extensive region, as a Haar-random state.
More precisely, for any bipartition $X\cup Y$ of the region reached by
the light cone, and any product state $\ket{s_Y}$ on $Y$, the
projected vector $\braket{s_Y}{\phi_t}$
has squared norm $\sim q^{-|Y|}$ to leading exponential accuracy, and
after normalization defines a Haar-typical state on $X$. Applying this
property recursively with $\ket{s_Y}=\ket{0_Y}$ yields the scaling of
$\mu_{2^{-j}}$ and the maximally mixed form of $\omega_{2^{-j}}$ used
below. We now derive the critical Rényi indices $\alpha_{c,j}$ for the nested hierarchy under this assumption. 

Let us first recall how the hierarchy is defined. At the top level, we decompose
\begin{equation}
O(t)\ket{0} = \ket{x}_A\otimes \ket{0_B}+\ket{y}_{AB},
\qquad
\braket{0_B}{y}=0. \label{eq:Ot0}
\end{equation}
where $|A| = |B| = N/2$. This leads to Eq.~\eqref{eq:spike-cloud-app}, with $\ket{v}\propto \ket{0_A}+\epsilon\ket{x_A}$. At the next level, we regard $\ket{v}\equiv \ket{v_{1/2}}$ as a pure state on
$A \equiv A_{1/2}=A_{1/4}\cup B_{1/4}$ and trace out $B_{1/4}$, with $|A_{1/4}| = |B_{1/4}| = N/4$. To do so, we first decompose
\begin{equation}
\ket{x}\equiv \ket{x_{A_{1/2}}}
=
\ket{x_{A_{1/4}}}\otimes \ket{0_{B_{1/4}}}
+
\ket{y_{A_{1/2}}},
\qquad
\braket{0_{B_{1/4}}}{y_{A_{1/2}}}=0.
\end{equation}
This is directly analogous to the decomposition of $O(t)\ket{0}$ at the top level. Tracing out $B_{1/4}$ then gives
\begin{equation}
\rho_{A_{1/4}}
=
\tr_{B_{1/4}}\ket{v_{1/2}}\bra{v_{1/2}}
=
(1-\mu_{1/4})\ket{v_{1/4}}\!\bra{v_{1/4}}+\mu_{1/4}\,\omega_{1/4},
\end{equation}
where $\ket{v_{1/4}}\propto \ket{0_{A_{1/4}}}+\epsilon\ket{x_{A_{1/4}}}$,
\begin{equation}
\omega_{1/4}
=
\frac{\tr_{B_{1/4}}\ket{y_{A_{1/2}}}\bra{y_{A_{1/2}}}}{\braket{y_{A_{1/2}}}{y_{A_{1/2}}}}
\end{equation}
is a normalized density matrix, and $\mu_{1/4} = \|y_{A_{1/2}}\|^2$. Recall that the key point, as mentioned in the main text and discussed in more detail in Appendix~\ref{sec:entanglement_of_v},  is that $\mu_{1/4} = \|y_{A_{1/2}}\|^2$ will generically be exponentially small in the system size which will lead to $\alpha_c < 1$ for $\rho_{A_{1/4}}$. This is because $\|x\|^2 = \|x_{A_{1/4}}\|^2 + \|y_{A_{1/2}}\|^2$, and in Eq.\ref{eq:Ot0}, by our assumption, $\|x\|^2 = \|\langle 0_B|O(t)|0\rangle\|^2$ is exponentially small in the system size.

Due to the above self-similar structure, the hierarchy can now be defined iteratively.

At level $j$, one obtains
\begin{equation}
\rho_{A_{2^{-j}}}
=
(1-\mu_{2^{-j}})\ket{v_{2^{-j}}}\!\bra{v_{2^{-j}}}
+
\mu_{2^{-j}}\,\omega_{2^{-j}},
\label{eq:spike-cloud-level-j}
\end{equation}
where $|A_{2^{-j}}|=N/2^j$ and $N$ is the total number of qudits with local Hilbert-space dimension $q$.
To determine the entropies, we need the scaling of $\mu_{2^{-j}}$ with system size. Under the aforementioned assumption, $\mu_{2^{-j}}$ has the same scaling as $\|x_{A_{2^{-(j-1)}}}\|^2$, where
\begin{equation}
\|x_{A_{2^{-(j-1)}}}\|^2
\sim \prod_{r = 1}^{j-1} q^{-N/2^r}
=
q^{-N(1-2^{-(j-1)})}.
\end{equation}
Although the product above is defined only for \(j > 1\), the final expression is valid also for $j=1$.

Indeed, at each step we write
\begin{equation}
\ket{x_{A_{2^{-(j-1)}}}}
=
\ket{x_{A_{2^{-j}}}}\otimes \ket{0_{B_{2^{-j}}}}
+
\ket{y_{A_{2^{-(j-1)}}}},
\qquad
\braket{0_{B_{2^{-j}}}}{y_{A_{2^{-(j-1)}}}}=0,
\end{equation}
and conditioning on region $B_{2^{-j}}$ contributes a factor of $q^{-N/2^j}$.

Further, under the same assumption as stated in the first paragraph of this subsection, the state $\omega_{2^{-j}}$ is approximately the identity density matrix on $N/2^j$ qudits, and therefore
\begin{equation}
\tr(\omega_{2^{-j}}^\alpha)\sim q^{(1-\alpha)N/2^j}.
\end{equation}
Using Eq.~\eqref{eq:spike-cloud-level-j}, the $0<\alpha<1$ Rényi moment is then controlled by
\begin{equation}
\tr(\rho_{A_{2^{-j}}}^\alpha)
\sim
c_1+c_2\,\mu_{2^{-j}}^\alpha \tr(\omega_{2^{-j}}^\alpha)
\sim
c_1+c_2\,
q^{-\alpha N(1-2^{-(j-1)})}\,
q^{(1-\alpha)N/2^j}
=
c_1+c_2\,
q^{\frac{N}{2^j}\left[1-(2^j-1)\alpha\right]},
\end{equation}
where $c_1,c_2=O(1)$. Hence the second term is exponentially large for $\alpha<1/(2^j-1)$ and exponentially small for $\alpha>1/(2^j-1)$, which yields
\begin{equation}
\alpha_{c,j}=\frac{1}{2^j-1}.
\end{equation}
Thus $\alpha_{c,1}=1$ for the top level, as expected from the discussion in the main text, while the next two levels give $\alpha_{c,2}=1/3$ and $\alpha_{c,3}=1/7$.

\section{Robustness of Rényi transitions across circuit ensembles}

In the main text we used Haar-random two-qubit gates as a convenient strongly scrambling circuit ensemble. To demonstrate explicitly that the observed hierarchy is not tied to Haar architecture, we also studied two other circuit ensembles described below. The first of them is a nearest-neighbor circuit built from a finite universal gate set, while the second one is a circuit that has a $U(1)$ symmetry.

\subsection{Hierarchical Rényi entropies for a circuit made of finite universal gate-set}
\label{sec:universal-gateset-circuit}

\begin{figure}[h]
\centering
\includegraphics[width=\columnwidth]{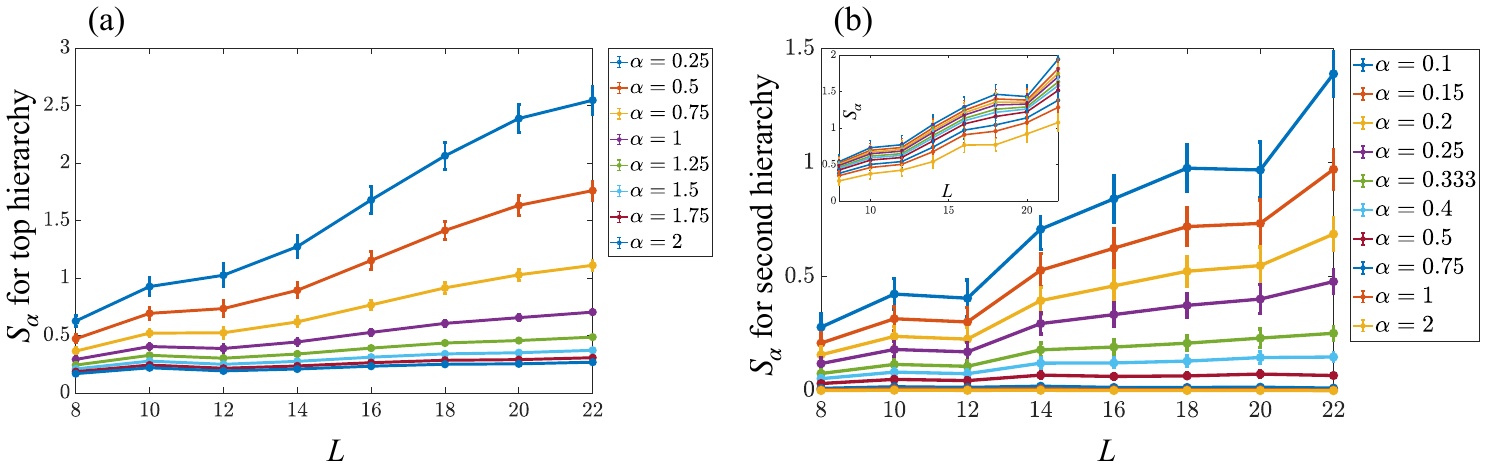}
\caption{Saturated Rényi entropies in the random circuit model with Clifford+T architecture, with $L_A = L/2, L_{A_{1/4}} = \lfloor L/4 \rfloor$, at $\epsilon = 0.1$.  The plotted entropies are averaged over thirty samples for both (a) and (b). The inset of (b) shows the Rényi entropies of the second highest Schmidt state in the second hierarchy. }
\label{fig:cliffordT_renyis}
\end{figure}

The circuit acts on a one-dimensional chain of \(L\) qubits with open boundary conditions. As in the Haar circuit, we study
\begin{equation}
|\psi(t)\rangle
=
\frac{(1+\epsilon O(t))|0\rangle}{\sqrt{\mathcal N_t}},
\qquad
O(t)=U(t)OU^\dagger(t),
\qquad
O=Z_{L/2}.
\end{equation}
The unitary \(U(t)\) is a depth-\(t\) random Clifford+\(T\) brickwork circuit. One timestep consists of two layers of single-site gates interleaved with two layers of nearest-neighbor CZ gates. More explicitly, for each timestep we choose with probability \(1/2\) whether the odd or even CZ layer is applied first. If \(\sigma_t\in\{\mathrm{odd},\mathrm{even}\}\) denotes this choice and \(\bar\sigma_t\) the complementary parity, then
\begin{equation}
U_t
=
U_{\mathrm{CZ}}^{\bar\sigma_t}
\,R^{(2)}_t\,
U_{\mathrm{CZ}}^{\sigma_t}
\,R^{(1)}_t .
\end{equation}
Here \(U_{\mathrm{CZ}}^{\mathrm{odd}}\) applies CZ gates on bonds
\((1,2),(3,4),\ldots\), while \(U_{\mathrm{CZ}}^{\mathrm{even}}\) applies CZ gates on bonds
\((2,3),(4,5),\ldots\). The random choice of \(\sigma_t\) helps in restoring statistical translational invariance.

Each single-site layer \(R^{(a)}_t\) is a product of independent one-qubit gates. On every site we first choose a gate uniformly from the finite Clifford subset $\{\mathds 1,H,S,S^\dagger,X,Z\}$
and then independently apply a non-Clifford gate with probability \(p_T\). When a non-Clifford gate is applied, it is chosen with equal probability from
\begin{equation}
T=\begin{pmatrix}1&0\\0&e^{i\pi/4}\end{pmatrix},
\qquad
T^\dagger=\begin{pmatrix}1&0\\0&e^{-i\pi/4}\end{pmatrix}.
\end{equation}
In the numerics we use $p_T=\frac12$.
For any fixed \(p_T>0\), this gives a local circuit drawn from a finite universal gate set. The case \(p_T=0\) reduces to a Clifford circuit.

We computed the same diagnostics as in the Haar-random circuit ensemble. First, we formed
\(\rho_{A_{1/2}}\) for the half-chain cut and computed its Rényi entropies. Second, we extracted the two largest Schmidt vectors of \(\rho_{A_{1/2}}\), bipartitioned the corresponding half-chain state into \(A_{1/4}|B_{1/4}\), and computed the resulting nested Rényi entropies. The results are shown in Fig.~\ref{fig:cliffordT_renyis}. They are qualitatively the same as for the Haar-random circuit: the full state (i.e., the top hierarchy) exhibits an area-to-volume-law transition at \(\alpha_c=1\), while the leading Schmidt state (i.e., the top state in the second hierarchy) exhibits a transition at \(\alpha_c<1\). The second Schmidt state in the second hierarchy instead shows volume-law behavior for all Rényi indices studied, consistent with it belonging to the  scrambled state $\omega$.

\subsection{Hierarchical Rényi entropies in a U(1) symmetric circuit} \label{sec:U1circuit}

\begin{figure}[h]
\centering
\includegraphics[width=\textwidth]{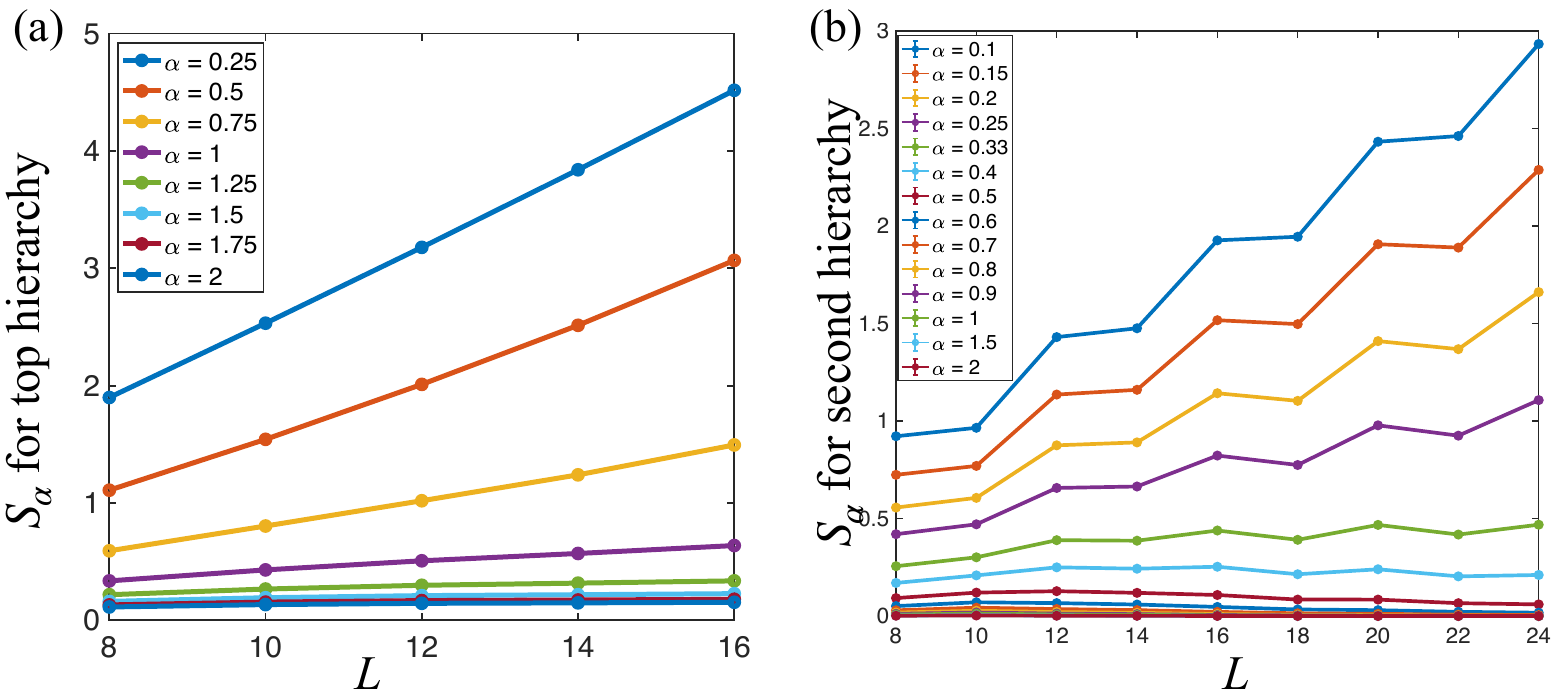}
\caption{Saturated Rényi entropies of the top Schmidt state corresponding to the top two hierarchies for the random circuit model with U(1) symmetry discussed in Appendix \ref{sec:U1circuit} with $\epsilon = 0.4$, $L_A = L/2, L_{A_{1/4}} = \lfloor L/4 \rfloor$. The plots are obtained by averaging over ten instances of the random circuit.}
\label{fig:U1renyis}
\end{figure}

The circuit we consider is based on Ref.\cite{khemani2018operator}.
The circuit is composed of nearest-neighbor two-site gates in the brickwork structure (i.e. one first applies gates on all odd bonds and then on even bonds). Each two-site gate conserves $\sum_i Z_i$. In the basis $\{\ket{00},\ket{01},\ket{10},\ket{11}\}$,
the gate has the block form
\begin{equation}
G =
e^{i\theta_{00}} \ket{00}\bra{00}
+
e^{i\theta_{11}} \ket{11}\bra{11}
+
\sum_{\alpha,\beta\in\{01,10\}}
V_{\alpha\beta}\ket{\alpha}\bra{\beta},
\label{eq:u1gate}
\end{equation}
where \(\theta_{00},\theta_{11}\) are random phases and \(V\in U(2)\) is Haar-random.

Similar to the case without symmetry, for a given $U(1)$ preserving circuit $U(t)$, we first define the Heisenberg operator $ O(t)=U(t)\, Z_{L/2}\, U^\dagger(t)$, where $Z_{L/2}$ is the
Pauli $Z$ operator on site $L/2$, and using this, we define the state of our main interest:
\begin{equation}
\ket{\psi_\Phi(t)}
=
\frac{(\I + \eps O(t))\ket{\Phi}}{\sqrt{\cN_t}},
\qquad
\cN_t = 1+\eps^2 + 2\eps a_t,
\label{eq:psit}
\end{equation}
where $\ket{\Phi}$ is a product state, and $a_t = \bra{\Phi} O(t)\ket{\Phi}$. In addition to the randomness in the circuit due to $U(2)$ Haar-random gates $V$, we also allow initial state $|\Phi\rangle$ to have randomness in the following sense: the center site is fixed in the \(+1\) eigenstate of \(Z_{L/2}\), namely \(\ket{0}\), and every other site is chosen independently from the \(X\)-basis states \(\ket{+},\ket{-}\):
\[
\ket{\Phi}
=
\ket{0}_{L/2} \otimes \bigotimes_{j\neq L/2} \ket{\eta_j}_j,
\qquad
 \eta_j\in\{+,-\}.
\] One advantage of this ensemble of initial states is that the ensemble average of $a_t$ is precisely the infinite-temperature autocorrelator of $Z_{L/2}$ which is expected to show diffusive behavior: \[
\overline{a_t}
=
\overline{\bra{\Phi}Z_{L/2}(t)\ket{\Phi}}
=
\frac{1}{2^L}\tr\!\big[(\mathbbm{1}+Z_{L/2})Z_{L/2}(t)\big]
=
\frac{1}{2^L}\tr\!\big[Z_{L/2} Z_{L/2}(t)\big].
\] 

The results for the Rényi entropies for the first two hierarchies are shown in Fig.\ref{fig:U1renyis}. They are again consistent with $\alpha_c = 1$ for the top hierarchy, and $\alpha_c < 1$ for the second hierarchy.

\section{Area-law Rényi entropies \texorpdfstring{$S_{\alpha>1}$}{S alpha>1} of $|\sqrt{\rho_{\beta,\theta}(t)}\rangle$ at  \texorpdfstring{$\beta,\theta=O(1)$}{beta,theta=O(1)}} \label{sec:renyisGibbs}

In this section we show that, assuming only that the Gibbs canonical purification  $|\sqrt{\rho_{\beta}}\rangle$ has area-law Rényi entropy $S_\alpha$ for Rényi-index $\alpha = 1/2$, the canonical purification  $|\sqrt{\rho_{\beta,\theta}(t)}\rangle$ of the quenched density matrix necessarily satisfies an area law of entanglement for Rényi indices $\alpha > 1$: \(S_{\alpha>1} \sim c L^{d-1}_A\) for all \(t\). 

The basic idea of the proof is that the overlap between $|\sqrt{\rho_{\beta}}\rangle$ and $|\sqrt{\rho_{\beta,\theta}(t)}\rangle$ remains non-zero and $O(1)$ for all times even in the thermodynamic limit, which implies that there always exists one large Schmidt eigenvalue of $|\sqrt{\rho_{\beta,\theta}(t)}\rangle$ for all times $t$. This large eigenvalue then leads to area-law Rényi entropy $S_{\alpha > 1}$ for the quenched state.

\subsection{$O(1)$ time-independent overlap between $|\sqrt{\rho_{\beta}}\rangle $ and $|\sqrt{\rho_{\beta,\theta}(t)}\rangle $}

Using the definition of the canonical purification, the overlap between two canonical purifications $|\sqrt{\rho_1}\rangle$ and $|\sqrt{\rho_2}\rangle$ is equal to $\tr( \sqrt{\rho_1} \sqrt{\rho_2})$. Therefore the overlap of  interest is:

\bea 
\langle \sqrt{\rho_{\beta}}|\sqrt{\rho_{\beta,\theta}(t)}\rangle
& = & \tr(\sqrt{\rho_{\beta}}\sqrt{\rho_{\beta,\theta}(t)}) \\
& = & \frac{1}{Z(\beta)}\tr\!\left(
e^{-\beta H/2} e^{i Ht} U^{\dagger}_\theta(x) e^{-\beta H/2} U_\theta(x) e^{-i Ht} \right) \\
& = & \frac{1}{Z(\beta)}
\tr\!\left(
e^{-\beta H/2} U_\theta^\dagger(x) e^{-\beta H/2}U_\theta(x)
\right).
\eea 

Note that the time-dependence completely drops out after taking trace. 

Due to $U_\theta(x)$ being a local unitary, one may write

\[
U_\theta^\dagger(x) H U_\theta(x) = H+\Delta H,
\qquad
\Delta H := U_\theta^\dagger(x) H U_\theta(x) - H.
\]

where \(\normop{\Delta H}=O(1)\).
Then
\be 
\langle \sqrt{\rho_{\beta}}|\sqrt{\rho_{\beta,\theta}(t)}\rangle
=
\frac{1}{Z(\beta)}
\tr\!\left(
e^{-\beta H/2} e^{-\beta(H+\Delta H)/2}
\right). \label{eq:overlap1}
\ee

\begin{proposition}
For any local unitary \(U_\theta(x)\),
\[
\langle \sqrt{\rho_{\beta}}|\sqrt{\rho_{\beta,\theta}(t)}\rangle \ge e^{-\beta \normop{\Delta H}/2}.
\]
Since \(\normop{\Delta H}=O(1)\), this implies that the overlap $\langle \sqrt{\rho_{\beta}}|\sqrt{\rho_{\beta,\theta}(t)}\rangle$ is bounded from below by an $O(1)$ time-independent constant for all time $t$. \label{prop:overlapbound}
\end{proposition}

\begin{proof}
Let's write the right-hand side of Eq.\ref{eq:overlap1} in the eigenbasis \(\{|n\rangle\}\) of \(H\), with \(H|n\rangle=E_n|n\rangle\). Then

\be 
\frac{1}{Z(\beta)}
\tr\!\left(
e^{-\beta H/2} e^{-\beta(H+\Delta H)/2}
\right) = \frac{1}{Z(\beta)} \sum_n e^{-\beta E_n/2}\,\langle n|e^{-\beta (H+\Delta H)/2}|n\rangle. \label{eq:overlap2}
\ee 

Using the convexity of the exponential, for any Hermitian \(X\) and normalized vector \(|\psi\rangle\),
\[
\langle \psi|e^X|\psi\rangle \ge e^{\langle \psi|X|\psi\rangle},
\]

Applying this to Eq.\ref{eq:overlap2} with $X=-\frac{\beta}{2}(H+\Delta H),\,\,|\psi\rangle=|n\rangle$,
gives
\[
\langle n|e^{-\beta (H+\Delta H)/2}|n\rangle
\ge
e^{-\beta\langle n|H+\Delta H|n\rangle/2}.
\]
Since
\[
\langle n|H+\Delta H|n\rangle
=
E_n+\langle n|\Delta H|n\rangle
\le
E_n+\normop{\Delta H},
\]
we obtain

\bea
\langle \sqrt{\rho_{\beta}}|\sqrt{\rho_{\beta,\theta}(t)}\rangle
& = & 
\frac{1}{Z(\beta)}
\tr\!\left(
e^{-\beta H/2} e^{-\beta(H+\Delta H)/2}
\right). \\
& = &  \frac{1}{Z(\beta)} \sum_n e^{-\beta E_n/2}\,\langle n|e^{-\beta (H+\Delta H)/2}|n\rangle\\
& \geq & \frac{1}{Z(\beta)} \sum_n e^{-\beta E_n/2} e^{-\beta(E_n+\normop{\Delta H})/2}\\
& = & e^{-\beta\normop{\Delta H}/2}.
\label{eq:overlap3}
\eea

\end{proof}
\begin{remark}
In the mixed-field Ising model used in the numerics in the main text, $H = \sum_{i=1}^{L-1} Z_i Z_{i+1}+ \sum_{i=1}^{L} (g X_i + h Z_i) + Z_1/4 - Z_L/4$, and \(U_\theta = e^{i \theta Z_{x_0}}\), one finds
\[
\Delta H = g\left(U_\theta^\dagger X_{x_0} U_\theta - X_{x_0}\right),
\]
A simple calculation yields
\[
\normop{\Delta H}
=
2|g|\left|\sin(\theta)\right|,
\]
and therefore
\[
\langle \sqrt{\rho_{\beta}}|\sqrt{\rho_{\beta,\theta}(t)}\rangle \ge \exp\!\left[-\beta |g| \left|\sin(\theta)\right|\right].
\]
\end{remark}

\subsection{Area-law Rényis \texorpdfstring{$S_{\alpha>1}$}{S alpha>1} for $|\sqrt{\rho_{\beta,\theta}(t)}\rangle$}

We now show that the time-independent overlap established in Eq.\ref{eq:overlap3} directly implies an area-law upper bound on the Rényi entropies $S_{\alpha>1}$ of $|\sqrt{\rho_{\beta,\theta}(t)}\rangle$. The argument is similar to that used in Ref.~\cite{muth2011dynamical} for operator entanglement, but we repeat it here for completeness in our present notation.

We bipartition the Hilbert space into the standard $A_sA_a|B_sB_a$. Let
\[
\ket{\phi}=\sum_i \sqrt{q_i}\,\ket{s_i}_A\ket{t_i}_B,\qquad
\ket{\psi}=\sum_j \sqrt{p_j}\,\ket{u_j}_A\ket{v_j}_B
\]
be the Schmidt decompositions of two bipartite pure states, with $q_1\ge q_2\ge\cdots$ and $p_1\ge p_2\ge\cdots$.

\begin{lemma}
For any $\alpha>1$,
\[
S_\alpha(\rho_A^\psi)
\le
S_{1/2}(\rho_A^\phi)
-\frac{2\alpha}{\alpha-1}\log |\braket{\phi}{\psi}|,
\]
where $\rho_A^\phi=\tr_B\ket{\phi}\bra{\phi}$ and $\rho_A^\psi=\tr_B\ket{\psi}\bra{\psi}$.
\end{lemma}

\begin{proof}
Writing $\ket{\phi},\ket{\psi}$ in a product basis on $A,B$ as
\[
\ket{\phi}=\sum_{ij}\Phi_{ij}\ket{a_i}\ket{b_j},
\qquad
\ket{\psi}=\sum_{ij}\Psi_{ij}\ket{a_i}\ket{b_j},
\]
their overlap is
\[
\braket{\phi}{\psi}=\tr(\Phi^\dagger\Psi).
\]
The singular values of $\Phi,\Psi$ are $\{\sqrt{q_i}\}$ and $\{\sqrt{p_i}\}$ respectively, so by the von Neumann trace inequality,
\[
|\braket{\phi}{\psi}|
=
|\tr(\Phi^\dagger\Psi)|
\le
\sum_i \sqrt{q_i p_i}.
\]
Let
\[
Q:=\sum_i \sqrt{q_i},
\qquad
w_i:=\frac{\sqrt{q_i}}{Q},
\qquad \sum_i w_i=1.
\]
Then
\[
|\braket{\phi}{\psi}|
\le
Q\sum_i w_i\,p_i^{1/2}.
\]
Since $\alpha>1$, the function $x^{1/\alpha}$ is concave on $x\ge 0$, and Jensen's inequality gives
\[
\sum_i w_i\,p_i^{1/2}
=
\sum_i w_i\,(p_i^{\alpha/2})^{1/\alpha}
\le
\left(\sum_i w_i\,p_i^{\alpha/2}\right)^{1/\alpha}.
\]
Therefore
\[
|\braket{\phi}{\psi}|
\le
Q\left(\sum_i w_i\,p_i^{\alpha/2}\right)^{1/\alpha}.
\]
Using the definition of $w_i$ and then Cauchy--Schwarz,
\[
\sum_i w_i\,p_i^{\alpha/2}
=
\frac{1}{Q}\sum_i \sqrt{q_i}\,p_i^{\alpha/2}
\le
\frac{1}{Q}\left(\sum_i q_i\right)^{1/2}\left(\sum_i p_i^\alpha\right)^{1/2}.
\]
Since $\sum_i q_i=1$, this becomes
\[
\sum_i w_i\,p_i^{\alpha/2}
\le
\frac{1}{Q}\left(\sum_i p_i^\alpha\right)^{1/2}.
\]
Substituting back, we obtain
\[
|\braket{\phi}{\psi}|
\le
Q^{1-1/\alpha}\left(\sum_i p_i^\alpha\right)^{1/(2\alpha)}.
\]
Now,
\[
Q^2=\left(\sum_i \sqrt{q_i}\right)^2=e^{S_{1/2}(\rho_A^\phi)},
\qquad
\sum_i p_i^\alpha=e^{(1-\alpha)S_\alpha(\rho_A^\psi)}.
\]
Hence
\[
|\braket{\phi}{\psi}|
\le
\exp\!\left[
\frac{\alpha-1}{2\alpha}S_{1/2}(\rho_A^\phi)
-\frac{\alpha-1}{2\alpha}S_\alpha(\rho_A^\psi)
\right].
\]
Taking logarithms and rearranging proves
\[
S_\alpha(\rho_A^\psi)
\le
S_{1/2}(\rho_A^\phi)
-\frac{2\alpha}{\alpha-1}\log |\braket{\phi}{\psi}|.
\]
\end{proof}

We now apply the Lemma to
\[
\ket{\phi}=|\sqrt{\rho_\beta}\rangle,
\qquad
\ket{\psi}=|\sqrt{\rho_{\beta,\theta}(t)}\rangle.
\]
Then
\[
\rho_A^\phi=\rho_{A,\beta},
\qquad
\rho_A^\psi=\rho_{A,\beta,\theta}(t),
\]
and therefore
\[
S_\alpha(\rho_{A,\beta,\theta}(t))
\le
S_{1/2}(\rho_{A,\beta})
-\frac{2\alpha}{\alpha-1}
\log\!\left|
\langle \sqrt{\rho_\beta}|\sqrt{\rho_{\beta,\theta}(t)}\rangle
\right|.
\]
Using the overlap bound from Eq.\ref{eq:overlap3},
\[
\left|
\langle \sqrt{\rho_\beta}|\sqrt{\rho_{\beta,\theta}(t)}\rangle
\right|
\ge e^{-\beta\normop{\Delta H}/2},
\]
we obtain
\[
S_\alpha(\rho_{A,\beta,\theta}(t))
\le
S_{1/2}(\rho_{A,\beta})
+
\frac{\alpha}{\alpha-1}\beta\normop{\Delta H}.
\]

We now assume the following:

\begin{assumption}[Area law for the Gibbs canonical purification]
For fixed \(\beta<\infty\), the reduced density matrix $\rho_{A,\beta} = \tr_B |\sqrt{\rho_\beta}\rangle \langle \sqrt{\rho_\beta}|$ for the canonical purification of the Gibbs state satisfies
\[
S_{1/2}(\rho_{A,\beta}) \le c L^{d-1}_A
\]
uniformly in \(L_A\) across the cut \(A_sA_a|B_sB_a\), where $c = O(1)$. \label{assume:shalfthermal}
\end{assumption}

\begin{corollary}
Assume that the Gibbs canonical purification satisfies Assumption \ref{assume:shalfthermal}, namely
\[
S_{1/2}(\rho_{A,\beta})\le cL_A^{d-1}.
\]
Then for every $\alpha>1$,
\[
S_\alpha(\rho_{A,\beta,\theta}(t))
\le
cL_A^{d-1}
+
\frac{\alpha}{\alpha-1}\beta\normop{\Delta H}.
\]
In particular, $|\sqrt{\rho_{\beta,\theta}(t)}\rangle$ satisfies an area law for all Rényi indices $\alpha>1$, uniformly in time.
\end{corollary}

\begin{corollary}[Large Schmidt eigenvalue]
Let $p_1$ denote the largest eigenvalue of
\[
\rho_{A,\beta,\theta}(t)=\tr_B |\sqrt{\rho_{\beta,\theta}(t)}\rangle\langle \sqrt{\rho_{\beta,\theta}(t)}|.
\]
Under Assumption \ref{assume:shalfthermal},
\[
p_1 \ge \exp\!\left[-cL_A^{d-1}-\beta\normop{\Delta H}\right].
\]
In particular, in one spatial dimension, for a fixed cut one has
\[
p_1 \ge e^{-c-\beta\normop{\Delta H}}=O(1),
\]
uniformly in time.
\end{corollary}

\begin{proof}
For any density matrix, the Rényi entropies are monotone decreasing in $\alpha$, and
\[
S_\infty=-\log p_1.
\]
Therefore, taking the limit $\alpha\to\infty$ in the preceding corollary gives
\[
S_\infty(\rho_{A,\beta,\theta}(t))
\le
cL_A^{d-1}+\beta\normop{\Delta H}.
\]
Since $S_\infty(\rho_{A,\beta,\theta}(t))=-\log p_1$, this implies
\[
p_1 \ge \exp\!\left[-cL_A^{d-1}-\beta\normop{\Delta H}\right].
\]
In $d=1$, $L_A^{d-1}=1$, so for a fixed cut the lower bound is $O(1)$.
\end{proof}

\begin{remark}
The preceding corollary shows that in $d=1$ the Schmidt spectrum of $|\sqrt{\rho_{\beta,\theta}(t)}\rangle$ necessarily contains a time-independent $O(1)$ leading weight, consistent with the numerics discussed in the main text.
\end{remark}

\section{Time dependence of Rényi entropies} \label{sec:timedepend}

\begin{figure}[h]
\centering
\includegraphics[width=\textwidth]{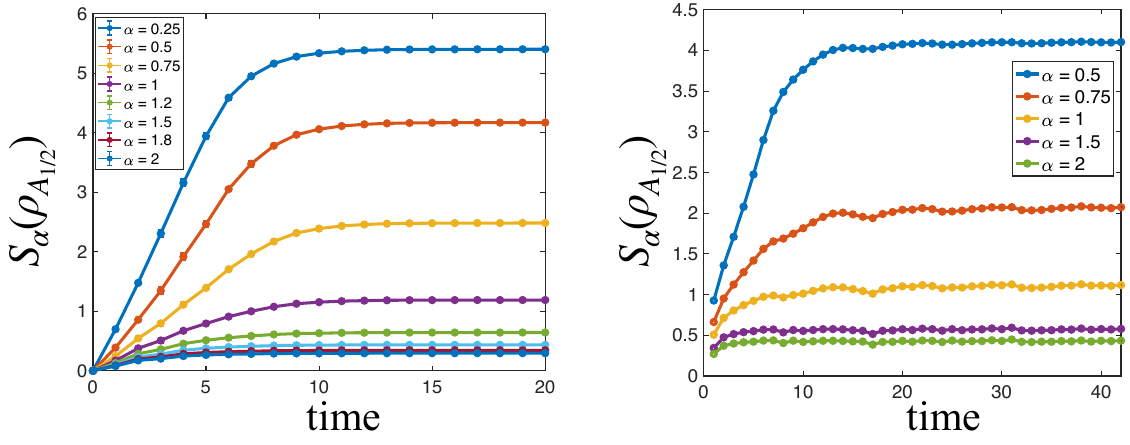}
\caption{(Left) Time dependence of the Rényi entropies for the state $\rho_{A_{1/2}}$ corresponding to the random circuit model at $\epsilon = 0.4$ (Eq.\ref{eq:randomcircuitstate}) for $L = 18, L_A = L/2$. The plot is for a single instance of the circuit. (Right) Time dependence of the Rényi entropies for the state $\rho_{A_{1/2}}$ corresponding to local quench on the Gibbs state of the mixed-field Ising model (Eq. \ref{eq:localquench}) at $g = 1.1, h = 0.35, \beta = 1, \theta = 0.5, L = 14, L_A = L/2$. }
\label{fig:Shalf_time_depend}
\end{figure}

Fig.\ref{fig:Shalf_time_depend} shows time dependence of the Rényi entropies of $\rho_{A_{1/2}}$ for the Haar-random circuit model as well as the locally quenched Gibbs state of the mixed-field Ising model discussed in the main text. In both cases, the entanglement grows monotonically before saturating (up to small fluctuations in the Gibbs case).

\section{Perturbative treatment of the locally quenched Gibbs state at high temperatures and arbitrary quench angle} 
\label{sec:smallbeta}

In this appendix we show that, in the regime $\beta \ll 1$ with $\theta = O(1)$, the canonical purification of the locally quenched Gibbs state reduces perturbatively to a product state across the cut plus a contribution from an operator that crosses the cut. This provides the precise connection to the circuit model discussed in the main text. We will only keep terms to linear order in $\beta$ to illustrate the idea.

We begin from
\begin{equation}
\rho_{\beta,\theta}(t)
=
\frac{e^{iHt}U_\theta^\dagger e^{-\beta H}U_\theta e^{-iHt}}{Z(\beta)},
\qquad
Z(\beta)=\tr(e^{-\beta H}),
\end{equation}
whose canonical purification is
\begin{equation}
|\sqrt{\rho_{\beta,\theta}(t)}\rangle
=
\frac{
e^{iHt}U_\theta^\dagger e^{-\beta H/2}U_\theta e^{-iHt}\otimes \Id_a
}{\sqrt{Z(\beta)}}|\Phi\rangle.
\end{equation}
Here $|\Phi\rangle = \sum_i |i\rangle_s |i\rangle_a$ is the unnormalized maximally entangled state between the system and ancilla. We again define
\begin{equation}
\Delta H:=U_\theta^\dagger H U_\theta - H,
\qquad
\Delta H(t):=e^{iHt}\Delta H\,e^{-iHt}.
\end{equation}
Then one has
\begin{equation}
e^{iHt}U_\theta^\dagger e^{-\beta H/2}U_\theta e^{-iHt}
=
e^{-\beta (H+\Delta H(t))/2}.
\end{equation}
Without loss of generality, we can assume $\tr(H)=0$ so that $Z(\beta)= |\mathcal{H}_S| +O(\beta^2)$ where $|\mathcal{H}_S|$ is the Hilbert-space dimension of the system. Therefore, perturbatively in $\beta$,
\begin{equation}
|\sqrt{\rho_{\beta,\theta}(t)}\rangle
=
\frac{|\Phi\rangle}{\sqrt{|\mathcal{H}_S|}}
-
\frac{\beta}{2\sqrt{|\mathcal{H}_S|}}(H+\Delta H(t))\otimes \Id_a\,|\Phi\rangle
+
O(\beta^2).
\label{eq:app_smallbeta_start}
\end{equation}

We now bipartition the doubled Hilbert space as $A_sA_a|B_sB_a$, with $A_s,B_s$ two contiguous halves of the physical system. Since $|\Phi\rangle$ factorizes across this cut,
\begin{equation}
|\Phi\rangle = |\Phi_A\rangle \otimes |\Phi_B\rangle,
\end{equation}
we also have
\begin{equation}
\frac{|\Phi\rangle}{\sqrt{|\mathcal{H}_S|}}
=
\frac{|\Phi_A\rangle}{\sqrt{|\mathcal{H}_{A_S}|}}
\otimes
\frac{|\Phi_B\rangle}{\sqrt{|\mathcal{H}_{B_S}|}},
\end{equation}
where $|\mathcal{H}_{A_S}|$ and $|\mathcal{H}_{B_S}|$ are the Hilbert-space dimensions of the system degrees of freedom in $A_s$ and $B_s$, respectively. It is useful to decompose the operator $H+\Delta H(t)$ into a part supported entirely in $A$, a part supported entirely in $B$, and a part with support on both sides of the cut:
\begin{equation}
H+\Delta H(t)=H_A(t)+H_B(t)+O_{\partial}(t).
\label{eq:app_decomp_operator}
\end{equation}
Here $H_A(t)$ and $H_B(t)$ contain all terms supported strictly within $A_s$ and $B_s$ respectively, while $O_{\partial}(t)$ contains all terms whose support intersects both sides of the cut.

A simple but important point is that the terms supported only in $A$ or only in $B$ do not generate new Schmidt structure at linear order in $\beta$. Indeed,
\begin{align}
\left[
1-\frac{\beta}{2}(H_A(t)+H_B(t))
\right]
\frac{|\Phi_A\rangle}{\sqrt{|\mathcal{H}_{A_S}|}}
\frac{|\Phi_B\rangle}{\sqrt{|\mathcal{H}_{B_S}|}}
&=
\frac{|\Phi_A\rangle}{\sqrt{|\mathcal{H}_{A_S}|}}
\frac{|\Phi_B\rangle}{\sqrt{|\mathcal{H}_{B_S}|}}
-\frac{\beta}{2}H_A(t)\frac{|\Phi_A\rangle}{\sqrt{|\mathcal{H}_{A_S}|}}
\frac{|\Phi_B\rangle}{\sqrt{|\mathcal{H}_{B_S}|}}
\nonumber\\
&\quad
-\frac{\beta}{2}
\frac{|\Phi_A\rangle}{\sqrt{|\mathcal{H}_{A_S}|}}
H_B(t)\frac{|\Phi_B\rangle}{\sqrt{|\mathcal{H}_{B_S}|}}
\nonumber\\
&=
\left(1-\frac{\beta}{2}H_A(t)\right)
\frac{|\Phi_A\rangle}{\sqrt{|\mathcal{H}_{A_S}|}}
\otimes
\left(1-\frac{\beta}{2}H_B(t)\right)
\frac{|\Phi_B\rangle}{\sqrt{|\mathcal{H}_{B_S}|}}
+O(\beta^2).
\label{eq:app_factorization}
\end{align}
Thus, up to corrections of order $\beta^2$, the action of $H_A(t)+H_B(t)$ merely dresses the product state $|\Phi_A\rangle |\Phi_B\rangle$.

Using Eqs.~\eqref{eq:app_smallbeta_start}--\eqref{eq:app_factorization}, one obtains
\begin{equation}
|\sqrt{\rho_{\beta,\theta}(t)}\rangle
=
|\phi_A(t)\rangle \otimes |\phi_B(t)\rangle
-
\frac{\beta}{2}O_{\partial}(t)\otimes \Id_a\,\frac{|\Phi\rangle}{\sqrt{|\mathcal{H}_S|}}
+
O(\beta^2),
\label{eq:app_main_formv1}
\end{equation}
where
\begin{equation}
|\phi_A(t)\rangle :=
\left(1-\frac{\beta}{2}H_A(t)\right)
\frac{|\Phi_A\rangle}{\sqrt{|\mathcal{H}_{A_S}|}},
\qquad
|\phi_B(t)\rangle :=
\left(1-\frac{\beta}{2}H_B(t)\right)
\frac{|\Phi_B\rangle}{\sqrt{|\mathcal{H}_{B_S}|}}.
\end{equation}
Since 
\[
\frac{|\Phi\rangle}{\sqrt{|\mathcal{H}_S|}}
=
|\phi_A(t)\rangle \otimes |\phi_B(t)\rangle
+
O(\beta),
\]
one may also write,
\begin{equation}
|\sqrt{\rho_{\beta,\theta}(t)}\rangle
=
\left(\Id -
\frac{\beta}{2}O_{\partial}(t)\otimes \Id_a\,\right)|\phi_A(t)\rangle \otimes |\phi_B(t)\rangle
+
O(\beta^2).
\label{eq:app_main_formv2}
\end{equation}

Eq.~\eqref{eq:app_main_formv2} is the precise analogue of the circuit-model state discussed in the main text, Eq.\ref{eq:randomcircuitstate}. The only difference is that the time-independent reference state $|0\rangle$ is replaced by the dressed, time-dependent product state $|\phi_A(t)\rangle \otimes |\phi_B(t)\rangle$, which differs from it only by terms supported locally within $A$ and $B$.

We now follow the same idea used in the circuit model in the main text and write
\begin{equation}
O_{\partial}(t)\otimes \Id_a\,\frac{|\Phi\rangle}{\sqrt{|\mathcal{H}_S|}}
=
|x(t)\rangle_A \otimes \frac{|\Phi_B\rangle}{\sqrt{|\mathcal{H}_{B_S}|}}
+
|y(t)\rangle,
\qquad
\left(
\Id_A\otimes \frac{\langle \Phi_B|}{\sqrt{|\mathcal{H}_{B_S}|}}
\right)
|y(t)\rangle = 0.
\label{eq:app_xydecomp_time}
\end{equation}
Since $|\phi_B(t)\rangle=|\Phi_B\rangle/\sqrt{|\mathcal{H}_{B_S}|}+O(\beta)$, substituting Eq.~\eqref{eq:app_xydecomp_time} into Eq.~\eqref{eq:app_main_formv2} gives, up to $O(\beta^2)$ corrections,
\begin{equation}
|\sqrt{\rho_{\beta,\theta}(t)}\rangle
=
\left(|\phi_A(t)\rangle-\frac{\beta}{2}|x(t)\rangle_A\right)
\otimes |\phi_B(t)\rangle
-
\frac{\beta}{2}|y(t)\rangle
+
O(\beta^2).
\end{equation}
Tracing out region $B$ one obtains
\begin{equation}
\rho_A(t)
=
\tr_B |\sqrt{\rho_{\beta,\theta}(t)}\rangle\langle \sqrt{\rho_{\beta,\theta}(t)}|
=
|v(t)\rangle\langle v(t)|
+ O(\beta^2),
\label{eq:app_rhoA_spikecloud}
\end{equation}
where
\begin{equation}
|v(t)\rangle
=
|\phi_A(t)\rangle - \frac{\beta}{2}|x(t)\rangle_A + O(\beta^2).
\end{equation}

Let's write 

\begin{equation}
\rho_A(t)=|v(t)\rangle\langle v(t)|+R_A(t),
\qquad
\|R_A(t)\|_1=O(\beta^2).
\label{eq:app_smallbeta_trace_norm}
\end{equation}

To make entropy-density statements we now assume that $\|R_A(t)\|_1\le C(t)\beta^2$ where $C(t)$ is independent of the subsystem size. This implies that  $T_A(t):=\frac12\|\rho_A(t)-|v(t)\rangle\langle v(t)|\|_1 \le \frac{C(t)}{2}\beta^2$. The Fannes--Audenaert continuity bound~\cite{fannes1973,audenaert2007sharp} then implies 
\begin{equation}
S_1(\rho_A(t))
\le
T_A(t)\log(\dim\mathcal H_A-1)+h_2(T_A(t)),
\label{eq:app_smallbeta_fannes}
\end{equation}
where $h_2(x)=-x\log x-(1-x)\log(1-x)$. Since the doubled region $A=A_sA_a$ has local Hilbert-space dimension $q^2$ per physical site, $\log\dim\mathcal H_A=2|A_s|\log q$. Therefore, the volume law coefficient of $S_1$ is bounded as

\begin{equation}
s_1 = \lim_{|A_s|\to\infty}
\frac{S_1(\rho_A(t))}{|A_s|}
\le
C(t)\beta^2\log q.
\label{eq:app_smallbeta_entropy_density_bound}
\end{equation}
Thus the volume-law coefficient of the top-level von Neumann entropy vanishes at least quadratically in $\beta$.

\section{Perturbative treatment of the locally quenched Gibbs state at small quench angle and $\beta = O(1)$} 
\label{sec:smalltheta}

In this appendix we consider the regime $\theta \ll 1$ at fixed $\beta = O(1)$. We show that, upon approximating the Gibbs canonical purification by a finite-bond-dimension state, the locally quenched purification takes a form directly analogous to the circuit-model state discussed in the main text.

Let
\begin{equation}
U_\theta = e^{i\theta Q_x},
\end{equation}
where $Q_x$ is a local Hermitian operator. The canonical purification of the quenched Gibbs state is
\begin{equation}
|\sqrt{\rho_{\beta,\theta}(t)}\rangle
=
\frac{
e^{iHt}U_\theta^\dagger e^{-\beta H/2}U_\theta e^{-iHt}\otimes \Id_a
}{\sqrt{Z(\beta)}}|\Phi\rangle.
\end{equation}
Expanding to linear order in $\theta$,
\begin{equation}
U_\theta^\dagger e^{-\beta H/2}U_\theta
=
e^{-\beta H/2}
+
i\theta [e^{-\beta H/2},Q_x]
+
O(\theta^2).
\end{equation}
Since $e^{-\beta H/2}$ commutes with $H$, this implies
\begin{equation}
|\sqrt{\rho_{\beta,\theta}(t)}\rangle
=
|\sqrt{\rho_\beta}\rangle
+
\frac{i\theta}{\sqrt{Z(\beta)}}\ket{[e^{-\beta H/2},Q_x(t)]}
+
O(\theta^2),
\label{eq:app_smalltheta_start}
\end{equation}
where
\begin{equation}
Q_x(t)=e^{iHt}Q_x e^{-iHt}.
\end{equation}
Equivalently, using the vectorization identities
\begin{equation}
\ket{AQ}=(\Id\otimes Q^T)\ket{A},
\qquad
\ket{QA}=(Q\otimes \Id)\ket{A},
\end{equation}
Eq.~\eqref{eq:app_smalltheta_start} may be written as
\begin{equation}
|\sqrt{\rho_{\beta,\theta}(t)}\rangle
=
|\sqrt{\rho_\beta}\rangle
+
i\theta\,
\bigl(\Id\otimes Q_x(t)^T-Q_x(t)\otimes \Id\bigr)
|\sqrt{\rho_\beta}\rangle
+
O(\theta^2).
\label{eq:app_smalltheta_exact}
\end{equation}

We now assume that the Gibbs canonical purification is well approximated by a finite-bond-dimension state across the cut $A_sA_a|B_sB_a$:
\begin{equation}
|\sqrt{\rho_\beta}\rangle
=
\sum_{i=1}^{D}\sqrt{p_i}\,|\Phi_A^i\rangle |\Phi_B^i\rangle,
\label{eq:app_thermal_approx}
\end{equation}
with $D=O(1)$.

Similar to the circuit model in the main text, we now decompose the term linear in $\theta$ in Eq.~\eqref{eq:app_smalltheta_exact} into the part that remains inside the Schmidt sector of $|\sqrt{\rho_\beta}\rangle$ on the $B$ side, and the orthogonal remainder:
\begin{equation}
\bigl(\Id\otimes Q_x(t)^T-Q_x(t)\otimes \Id\bigr)\sum_{i=1}^{D}\sqrt{p_i}\,|\Phi_A^i\rangle |\Phi_B^i\rangle
=
\sum_{i=1}^{D}|x_A^i(t)\rangle |\Phi_B^i\rangle
+
|y(t)\rangle,
\label{eq:app_smalltheta_decomp}
\end{equation}
with $\langle \Phi_B^i|y(t)\rangle=0$
for all $i$. Substituting Eq.~\eqref{eq:app_smalltheta_decomp} into Eq.~\eqref{eq:app_smalltheta_exact}, we find
\begin{equation}
|\sqrt{\rho_{\beta,\theta}(t)}\rangle
=
\sum_{i=1}^{D}
\left(
\sqrt{p_i}\,|\Phi_A^i\rangle+i\theta |x_A^i(t)\rangle
\right)|\Phi_B^i\rangle
+i\theta |y(t)\rangle
+
O(\theta^2).
\label{eq:app_smalltheta_finalstate}
\end{equation}

Since $|y(t)\rangle$ is orthogonal to all  $|\Phi_B^i\rangle$, tracing out region $B$ leads to the following expression for the density matrix on $A$:
\begin{equation}
\rho_A(t)
=
\tr_B |\sqrt{\rho_{\beta,\theta}(t)}\rangle\langle \sqrt{\rho_{\beta,\theta}(t)}|
=
\sum_{i=1}^{D}
\left(
\sqrt{p_i}\,|\Phi_A^i\rangle+i\theta |x_A^i(t)\rangle
\right)
\left(
\sqrt{p_i}\,\langle \Phi_A^i|-i\theta \langle x_A^i(t)|
\right)
+
O(\theta^2).
\label{eq:app_smalltheta_finalrhoA}
\end{equation}

Thus, to linear order in $\theta$, the reduced density matrix is controlled entirely by the same number of $O(1)$ Schmidt states as the finite-bond-dimension approximation to the thermal purification. In this sense, the small-$\theta$ expansion has the same structure as the circuit-model state discussed in the main text, with the one-dimensional sector replaced by the $D$-dimensional Schmidt sector of $|\sqrt{\rho_\beta}\rangle$.

Following the same argument as in Appendix~\ref{sec:smallbeta}, if one again assumes that the one-norm of the $O(\theta^2)$ correction in Eq.~\eqref{eq:app_smalltheta_finalrhoA} is uniformly bounded by a system-size independent constant, then  the Fannes--Audenaert continuity bound~\cite{fannes1973,audenaert2007sharp}  implies that the volume law coefficient of $S_1(\rho_A)$ vanishes at least quadratically in $\theta$.

\section{Numerical results in the small $\beta, \theta$ regime for the canonical purification $|\sqrt{\rho_{\beta,\theta}(t)}\rangle$} 

\subsection{Scaling of volume law coefficient of von Neumann entropy at small $\beta, \theta$} \label{sec:volume-law-coeff}

\begin{figure}[h]
\centering
\includegraphics[width=\textwidth]{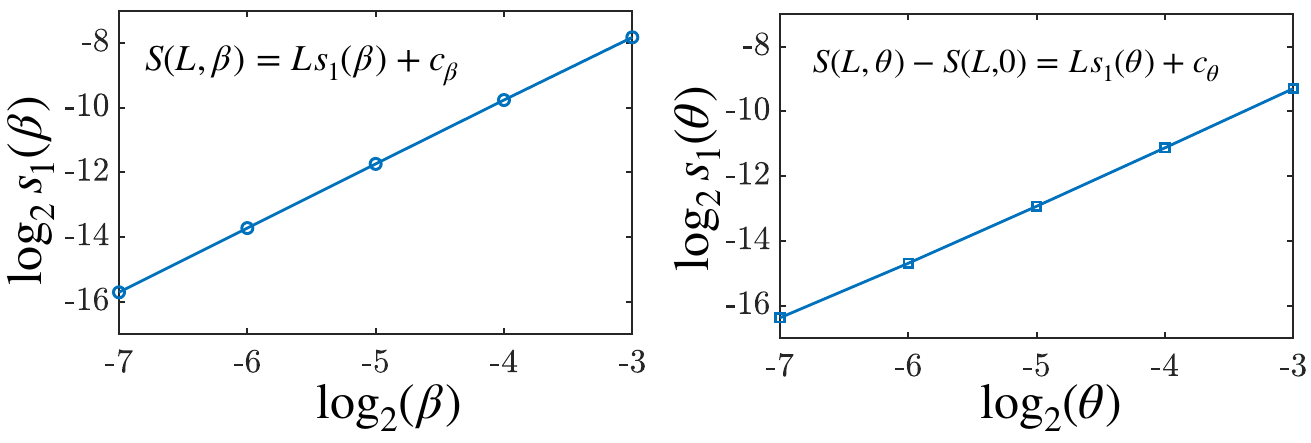}
\caption{(Left)  Scaling of the volume-law coefficient of the bipartite von Neumann entropy corresponding to a half-half bipartition of the canonical purification $|\sqrt{\rho_{\beta,\theta}(t)}\rangle$ at time $t = L^2$ when $S_1$ has plateaued to its maximum value. Here $\theta = 1$. (Right)   Scaling of the volume-law coefficient of the bipartite von Neumann entropy corresponding to a half-half bipartition of the canonical purification $|\sqrt{\rho_{\beta,\theta}(t)}\rangle$ at time $t = L^2$ when $S_1$ has plateaued to its maximum value. Here $\beta = 0.8$.}
\label{fig:s1_slope_smallbetatheta}
\end{figure}

To extract the volume-law coefficient of the saturated von Neumann entropy of  $|\sqrt{\rho_{\beta,\theta}(t)}\rangle$, we performed a standard least-squares fit in system size at each fixed value of the scanned parameter. For the small $\beta$ regime with $\theta = O(1)$, the data points \(S_1(L,\beta)\) at \(L=8,10,12,14\) were fit to
\[
S_1(L,\beta)=L\,s_1(\beta)+c_\beta,
\]
where our main object of  interest is the dependence of \(s_1(\beta)\), the entropy density, on $\beta$. For the small $\theta$ regime, we first subtracted the \(\theta=0\) (static) Gibbs contribution at the same system size and fit the excess entropy according to
\[
S_1(L,\theta)-S_1(L,0)=L\,s_1(\theta)+c_\theta,
\]

again for  \(L=8,10,12,14\). Fig.\ref{fig:s1_slope_smallbetatheta} (a) shows \(\log_2 s_1(\beta)\) versus \(\log_2 \beta\), and  Fig.\ref{fig:s1_slope_smallbetatheta} (b) shows \(\log_2 s_1(\theta)\) versus \(\log_2 \theta\). We find $s_1(\beta) \sim \beta^{1.97}$ and $s_1(\theta) \sim \theta^{1.77}$. The $\beta$ exponent is very close to the expected quadratic scaling, while the one for $\theta$ is broadly consistent with \(s_1(\theta)\sim\theta^2\), with the deviation plausibly due to finite-size effects.

\subsection{Scaling of truncation error when both $\beta, \theta$ are small} \label{sec:beta_theta_small}

\begin{figure}[h]
\centering
\includegraphics[width=\textwidth]{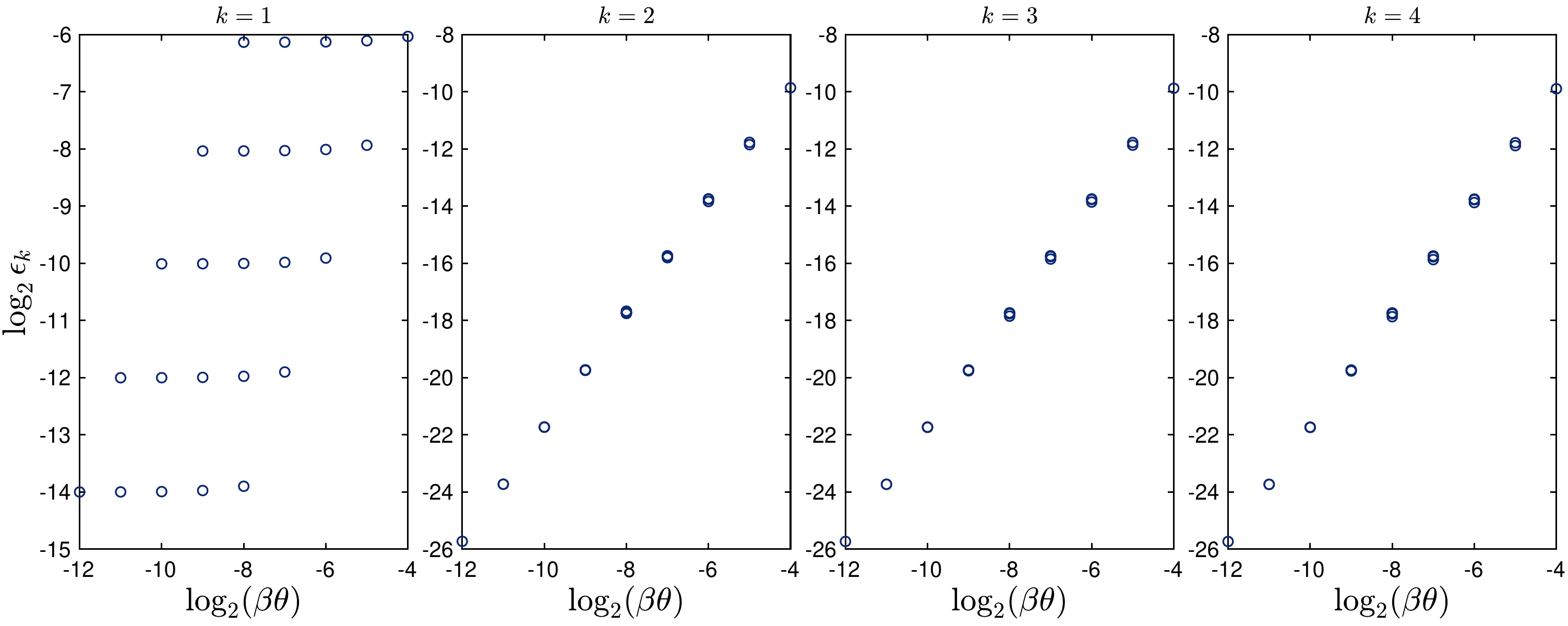}
\caption{Truncation error when $\beta, \theta \ll 1$ for $L=12$. The data are taken over a $5\times 5$ grid $\beta,\theta\in\{2^{-2},2^{-3},2^{-4},2^{-5},2^{-6}\}$. We plot $\log_2 \epsilon_k$ versus $\log_2(\beta\theta)$ for $k=1,2,3,4$, where $\epsilon_k=1-\sum_{j=1}^k p_j$ and $p_j$ are the Schmidt weights across the $A_sA_a|B_sB_a$ cut. For $k>1$, the data collapse onto straight lines with slope $\approx 2.0$, consistent with theoretical expectation $\epsilon_k\propto(\beta\theta)^2$.}
\label{fig:data_collapse}
\end{figure}

Fig.\ref{fig:z_string}(b) shows that when $\beta \ll 1$ and $\theta = O(1)$, the truncation error vanishes as $\beta^2$, and similarly, when $\theta \ll 1$ and $\beta = O(1)$, the truncation error vanishes as $\theta^2$. For completeness, we now show data when both $\beta, \theta \ll 1$. We fix the system size to $L = 12$. For each pair $(\beta,\theta)$, we construct $|\sqrt{\rho_{\beta,\theta}(t)}\rangle$
and compute its Schmidt weights $\{p_j\}$ across the spatial bipartition $A_sA_a|B_sB_a$. The truncation error after keeping the largest $k$ Schmidt states is $\epsilon_k = 1-\sum_{j=1}^{k}p_j $. We scan a grid of small values of $\beta$ and $\theta$ and plot $\epsilon_k$ as a function of the product $\beta\theta$, see Fig.~\ref{fig:data_collapse}. The resulting data show a clear collapse for $k>1$:
\[
\epsilon_k \sim C_k(\beta\theta)^2,
\]
consistent with our expectation.

\end{document}